\documentclass[a4paper,11pt]{article}
 \usepackage[scale={0.75,0.9},centering,includeheadfoot]{geometry}
  \usepackage{algorithm}
  \usepackage{algpseudocode}

   \usepackage{subcaption}

 \usepackage{placeins}
 \usepackage{array}
 \usepackage[latin1]{inputenc} 
 \usepackage[british]{babel}
 \usepackage{latexsym}
 \usepackage[dvips]{graphics,graphicx,psfrag}
 \usepackage{amssymb,amsmath,amsfonts,amsthm}
 \usepackage{verbatim}
 \usepackage{pstricks,pst-node,psfrag}
 \usepackage[hyphens]{url}
\usepackage[breaklinks,hyperindex]{hyperref}
 \usepackage{epsfig}
 \usepackage{setspace}
 \usepackage{booktabs}

 \usepackage{natbib}
 \graphicspath{{figs/}{./}{../figs/}}
\newtheorem{thm}{Theorem}[section]

\newtheorem{prop}[thm]{Proposition}
\newtheorem{theorem}[thm]{Theorem}

\theoremstyle{definition} 

\theoremstyle{remark}

\newcommand{\proper}{\mathsf}

\newcommand{\pE}{\proper{E}}

\newcommand{\pC}{\proper{C}}
\newcommand{\pN}{\proper{N}}
\newcommand{\argmin}{arg\,min}
\newcommand{\mv}[1]{{\boldsymbol{\mathrm{#1}}}}
\newcommand{\trsp}{\ensuremath{\top}}

\renewcommand{\O}{\mathcal{O}}

\newcommand{\V}{{\bf \proper{V}}}

\theoremstyle{remark}

\usepackage{fancyhdr}
\title{Efficient adaptive MCMC through precision estimation}
\author{Jonas Wallin}
\date{}

\begin{document}
\begin{center}
\textbf{{\huge
Efficient adaptive MCMC through precision estimation}}\\
\vspace{5mm}
{\Large \scshape{Jonas Wallin and David Bolin}}

\vspace{3mm}

\textit{Mathematical Sciences, Chalmers and University of Gothenburg}

\vspace{3mm}
\begin{minipage}{0.9\textwidth}
{\small
{\bf ABSTRACT.
A novel adaptive Markov chain Monte Carlo algorithm is presented.
The algorithm utilizes sparisity in the partial correlation structure of a density to efficiently estimate the covariance matrix through the Cholesky factor of the precision matrix. The algorithm also utilizes the sparsity to sample efficiently from both MALA and Metropolis Hasting random walk proposals. Further, an algorithm that estimates the partial correlation structure of a density is proposed. Combining this with the Cholesky factor estimation algorithm results in an efficient black-box AMCMC method that can be used for general densities with unknown dependency structure. The method is compared with regular empirical covariance adaption for two examples. In both examples, the proposed method's covariance estimates converge faster to the true covariance matrix and the computational cost for each iteration is lower.
}
\begin{flushleft}
\end{flushleft}
}
\end{minipage}

\vspace{3mm}
\end{center}


\section{Introduction}
Markov chain Monte Carlo (MCMC) algorithms are widely used for sampling from complicated distributions. As the dimension of the distribution gets larger, tuning the parameters of the algorithms becomes both more important and  more difficult. Adaptive MCMC (AMCMC) methods address this issue by tuning the MCMC algorithm while it is running.

One of the most popular MCMC algorithms is the Metropolis Hasting random walk (MHRW) algorithm  \citep{tierney1994}. In the MHRW, new samples are proposed using the proposal kernel
$
\mathcal{Q}(\mv{x},.) = \proper{N}( \mv{0}, \mv{\Sigma}),
$
where $\mv{\Sigma}$ is a scaling matrix and $\proper{N}$ denotes the normal distribution. The optimal scaling matrix $\mv{\Sigma}$ is a rescaled version of the covariance matrix of $\pi$. One of the main reasons for the popularity of the algorithm is that it is easy to implement; however, even the optimal scaling of the covariance matrix decreases as $\O(n^{-1})$ where $n$ is the dimension of $\mv{x}$ \citep{roberts2001}.  This means that optimal proposals are small for large dimensions, which causes poor mixing of the MCMC chain.

An alternative to the MHRW algorithm is the Metropolis adjusted Langevian algorithm (MALA) \citep{grenander1994representations}. The MALA is generally more difficult to implement compared with the MHRW, as it proposal kernel is
$\mathcal{Q}(\mv{x},.) = \proper{N}( \frac{\mv{\Sigma}}{2}\nabla \log(\pi(\mv{x})), \mv{\Sigma})$. The advantage is that the optimal scaling matrix decreases with $\O(n^{-1/3})$, which is a big improvement compared with the MHRW scaling.

Using the covariance matrix, of $\pi$, for the proposal has a huge effect on the convergence of the algorithms
in practice. Unfortunately, the matrix is almost never known in advance. In the AMCMC framework one solves this by replacing the covariance matrix with the empirical covariance matrix (ECM) from previous samples in the MCMC algorithm \citep{haario2001}.
If the dimension of $\pi$ is large, using the ECM is problematic: It converges very slowly towards the true covariance matrix \citep{EAMCMC_Roberts} and the evaluation of the proposal density requires computing the inverse of the ECM. 
In large dimensionon, the inverse will often be the computational bottle neck of MH algorithm.

The computation of the inverse is so computationally expensive that it often dominates the computational cost for each iteration of the MCMC algorithm.

In this paper, we propose an AMCMC algorithm that reduces these two issues for a large class of high dimensional densities. The general idea is to estimate the Cholesky factor of the precision (inverse covariance) matrix instead of using the ECM directly. The precision matrix and its Cholesky factor will be sparse if the target density has a sparse partial correlation structure \citep{lauritzen1996graphical, rue2005gaussian}. If the Cholesky factor is sparse there are fewer elements to estimate than in the covariance matrix, resulting in faster convergence. Also, the evaluation of the proposal density becomes significantly faster by using the sparse Cholesky factor.

The algorithm is especially well suited for Bayesian hierarchical models since they are often constructed by a directed acycle graph (DAG), which defines a conditional dependency structure of the model.

A potential issue with the method is that it cannot be used if the partial correlation structure is unknown. To remedy this, we propose a second algorithm which estimates the conditional density structure, and uses this to define the sparsity of the Cholesky factor. Although partial correlation does not imply conditional independence and vice versa, it is reasonable to use this as a sparsity pattern for the proposal distribution. Combining this algorithm with the Cholesky estimation algorithm results in an efficient black-box AMCMC method that can be used on general densities with unknown dependency structure.

The article is composed as follows. In Section 2, regular covariance-based AMCMC for MHRW and MALA is presented and the new adaption method is introduced. In Section 3, the algorithm that updates the Cholesky factor in an MCMC iteration is presented. The method is applied to two problems in Section 4. In both problems, the densities are derived from discretizations of infinite dimensional models, and has a lot of conditional dependency structure. Finally, Section 5 contains a brief discussion of future work.

The code for the method and the examples are available at \citep{code2015}.
\section{Adapation of Metropolis Hastings algorithms}
The classical adaptive MHRW for an $n$-dimensional target distribution $\pi$, introduced by \cite{haario2001}, uses at iteration $i$ the proposal
$\mathcal{Q}_i( \mv{x},.) = \proper{N} \left( \mv{x}, \frac{2.38^2}{n} \mv{\Sigma}^{(i)} \right)$.
Here
\begin{equation}\label{eq:empcov}
\mv{\Sigma}^{(i)}  = \frac{1}{i} \sum_{j=1}^i \left(\mv{x}^{(j)}  - \bar{\mv{x}}^{(i)} \right) \left(\mv{x}^{(j)}  - \bar{\mv{x}}^{(i)} \right)^{T},
\end{equation}
where $\bar{\mv{x}}^{(i)} = \frac{1}{i} \sum_{j=1}^i  \mv{x}^{(j)}$ and $\{ \mv{x}_j \}_{j=1}^i$ are all previous $i$ samples. Typically the ECM $\mv{\Sigma}^{(i)}$ is adjusted with $\mv{I}\epsilon$ to ensure that the matrix is positive definite.
Similar adaptation can be used for the MALA and the proposal distribution is in that case
$$
\mathcal{Q}_i(\mv{x},.) = \proper{N}\left( \mv{x} +  \frac{\sigma^2}{2} \mv{\Sigma}^{(i)} \nabla \log( \pi ( \mv{x})) ,\sigma^2  \mv{\Sigma}^{(i)} \right),
$$
where $\sigma$ is a suitable scaling factor, that often also needs to be adapted. We will refer to the method above as covariance adaption. Algorithm \ref{alg:MALA_ecm} display an iteration of MALA with covariance adaption.

 \begin{algorithm}[t]
 \caption{MALA with covariance-based adaptation}
 \label{alg:MALA_ecm}
 \begin{algorithmic}[1]
 \Procedure{MALA}{$ \mv{x}, \mv{\Sigma}, \bar{\mv{x}}, \sigma, \pi, i$}
\State $\mv{L}  \gets Chol( \mv{\Sigma})$
\State $\mv{z} \gets \proper{N}(\mv{0},\mv{I}) $
\State $\mv{g} \gets  \nabla \log( \pi ( \mv{x}))$
\State $\mv{g}_l \gets  \mv{L}^{\trsp}\mv{g}$
 \State $\mv{x}^* \gets  \frac{\sigma^2}{2}\mv{\Sigma}  \mv{g} + \sigma \mv{L} \mv{z}$
 \State $\mv{g}^* \gets  \nabla \log( \pi ( \mv{x}^* ))$
 \State $\mv{g}^*_l \gets \mv{L}^{\trsp}\mv{g}^*$
 \State $U \gets \mathcal{U}(0,1)$
 \State
$\begin{aligned}[t]
\alpha \gets \log&(\pi(\mv{x}^*)) - \log(\pi(\mv{x})) - \frac{\sigma^2}{8}  (\mv{g}_l^* )^{\trsp} \mv{g}^*_l  \\
  &+ \frac{1}{2} (\mv{x} -\mv{x}^*)^{\trsp} (\mv{g}^*  +  \mv{g}) + \frac{\sigma^2}{8}  (\mv{g}_l)^{\trsp} \mv{g}_l
 \end{aligned}$
 \If{$\log(U) < \alpha$ }
 \State $\mv{x} \gets \mv{x}^*$
 \EndIf
 \State  $ \bar{\mv{x}} \gets \frac{i}{i+1}  \bar{\mv{x}} +  \frac{1}{i+1}\mv{x}$
 \State  $\mv{\Sigma} \gets \frac{i}{i+1}  \mv{\Sigma} +  \frac{1}{i+1}   \left(\mv{x}  - \bar{\mv{x}} \right) \left(\mv{x} - \bar{\mv{x}} \right)^{T}$
 \State \Return $\{\mv{x},\bar{\mv{x}},\mv{\Sigma} \}$
 \EndProcedure
 \end{algorithmic}
 \end{algorithm}

A simple way to reduce the computation time of the algorithm is to update the Cholesky factor of the ECM directly instead of the covariance matrix at line $15$ in Algorithm \ref{alg:MALA_ecm}.

The adaption method we propose uses the same proposals, but instead of estimating the covariance matrix it estimates the Cholesky factor of the precision matrix, denoted $\mv{L}$.  We denote this adaption method precision adaption.  The proposal for the precision adaptation is easily reconstructed from the equations above by replacing $\mv{\Sigma}^{(i)}$ with $(\mv{L}^{(i)}(\mv{L}^{(i)})^{\trsp})^{-1}$.  Algorithm~\ref{alg:MALA_prec} displays an iteration of MALA with precision adaptation. Estimating $\mv{L}$ is not as straightforward as estimating the covariance matrix with the ECM, and the crucial L-Update method at line $14$ in the algorithm will be described in the next section.

 \begin{algorithm}[t]
 \caption{MALA with precision-based adaptation}
  \label{alg:MALA_prec}
 \begin{algorithmic}[1]
 \Procedure{MALA}{$ \mv{x}, \mv{L}, \bar{\mv{x}}, \sigma, \pi,i,...$}
\State $\mv{z} \gets \proper{N}(\mv{0},\mv{I}) $
\State $\mv{g} \gets   \nabla \log( \pi ( \mv{x}))$
\State $\mv{g}_l \gets   \mv{L}^{-1}\mv{g} $
 \State $\mv{x}^* \gets \mv{L}^{-T} \left(\mv{L}^{-1} \frac{\sigma^2}{2}   \mv{g} + \sigma \mv{z} \right)$
 \State $\mv{g}^* \gets  \nabla \log( \pi ( \mv{x}^* ))$
 \State $\mv{g}^*_l \gets    \mv{L}^{-1}\mv{g}^* $
 \State $U \gets \mathcal{U}(0,1)$
 \State $\begin{aligned}[t]
\alpha \gets  \log&(\pi(\mv{x}^*)) - \log(\pi(\mv{x})) - \frac{\sigma^2}{8}  (\mv{g}_l^* )^{\trsp} \mv{g}^*_l   \\
&+ \frac{1}{2} (\mv{x} -\mv{x}^*)^{\trsp} (\mv{g}^*  +  \mv{g}) + \frac{\sigma^2}{8}  (\mv{g}_l)^{\trsp} \mv{g}_l
 \end{aligned}$

 \If{$\log(U) < \alpha$ }
 \State $\mv{x} \gets \mv{x}^*$
 \EndIf
 \State  $ \bar{\mv{x}} \gets \frac{i}{i+1}  \bar{\mv{x}} +  \frac{1}{i+1}\mv{x}$
 \State  $\mv{L} \gets$ \textproc{L-update}($\mv{x}-\bar{\mv{x}}$,...) \Comment{see Algorithm \ref{alg:Lupdate} for details}
 \State \Return $\{\mv{x},\bar{\mv{x}},\mv{L},... \}$
 \EndProcedure
 \end{algorithmic}
 \end{algorithm}

When introducing adaptation in the algorithm it is important to ensure that $\pi$ remains the stationary distribution of the algorithm. In Appendix \ref{sec:app_conv}, we adress this by connecting convergence of precision adaption to covariance adaption.

\section{Online estimation of the Cholesky factor}
In this section we present a method for online estimation of the Cholesky factor of a density's precision matrix. The method will be formulated so that it can take advantage of sparsity properties of the Cholesky factor, which determines the partial correlation structure.

The section is structured as follows. In Section \ref{seq:uChol}, we derive an algorithm producing an online least squares estimate, $\mv{L}^{(i)}$, of the Cholesky factor from $i$ samples from $ \pi$, assuming that the partial correlation structure of $\pi$ is known. After this, in Section \ref{sec:SparP}, an algorithm for estimating the conditional dependency structure of $\pi$ is derived, which can be used if the partial correlation structure is not known.

\subsection{Least squares estimation} \label{seq:uChol}
Assume that  $\mv{X}$ is a random vector with distribution $\pi$ and finite positive definite covariance matrix $\mv{\Sigma}$. Also assume,  without loss of generality, that $\pE(\mv{X}) = \mv{0}$. We will in this section formulate an estimation method for the Cholesky factor, $\mv{L}$, of the precision matrix $\mv{Q} = \mv{\Sigma}^{-1}$. This will be done by first finding a representation of $\mv{L}$ as a series of regression problems, which is a well-known method to estimate covariance matrices for normal distributions \citep[see for example][]{pourahmadi2011covariance}.
The theorem below shows that regressing $\mv{X}_j$ on $\mv{X}_{j+1:N}$ for $j=1,\ldots, N$ can be used to construct $\mv{L}$. This result is well-known when $\mv{X}$ is Normally distributed; however, we have not been able to find a proof without the normal assumption so we provide one here.

\begin{theorem}
Let $\mv{X}$ be an $N-$dimensional random variable with zero mean and positive definite covariance matrix $\mv{\Sigma}$. Define $\mv{T}$ as the upper triangular matrix with unit diagonal and  $\mv{T}_{j,j+1:N}= -\mv{\Sigma}^{-1}_{j+1:N,j+1:N} \mv{\Sigma}_{j,j+1:N}$
for $j=1,\ldots,N$. Further, define $\mv{D}$ as the diagonal matrix with $\mv{D}_{jj}=\mv{\Sigma}_{jj} - \mv{\Sigma}^{\trsp}_{j,j+1:N} \mv{\Sigma}^{-1}_{j+1:N,j+1:N}\mv{\Sigma}_{j,j+1:N}$. Then the Cholesky factor of $\mv{\Sigma}^{-1}$ is $\mv{L} =  \mv{T}^{\trsp} \mv{D}^{-1/2}$,
  and $\mv{T}_{j,j+1:N}$ is the solution to
\begin{equation}\label{eq:mse}
\argmin_{\mv{t}} \scalebox{0.9}{$\pE\left[\left(X_j - \sum_{k=1}^{N-j} t_{k} X_k\right)^2\right]$},
\end{equation}
\end{theorem}

\begin{proof}
Let $\mv{\epsilon} \overset{d}{=} \mv{T}\mv{X}$ and $\mv{D} = \pC[\mv{\epsilon}]$. By construction, the variance of $\mv{\epsilon}_j$ is $\mv{D}_{jj}=\mv{\Sigma}_{jj} - \mv{\Sigma}^{\trsp}_{j,j+1:N} \mv{\Sigma}^{-1}_{j+1:N,j+1:N}\mv{\Sigma}_{j,j+1:N}$. Since $\mv{D} = \pC[\mv{\epsilon}] =  \pC[\mv{TX}] = \mv{T}^{\trsp}\mv{\Sigma}\mv{T}$, the proof is completed by can showing that $\mv{D}$ diagonal matrix. Assume that $\mv{D}$ is not a diagonal matrix and that $\mv{D}_{ij}$, $j<i$ is non zero. Thus
$$
 \scalebox{0.9}{$\V\left[\mv{\epsilon}_j - \frac{\mv{D}_{ij}}{\mv{D}_{ii}} \mv{\epsilon}_i\right] < \V[ \mv{\epsilon}_j ]$},
$$
or equivalently, by the definition of $\mv{\epsilon}$,
$$
\scalebox{0.88}{$\V\left[ \mv{X}_j  - \mv{T}_{j,j+1:N}^{\trsp} \mv{X}_{j+1:N}   - \frac{\mv{D}_{ij}}{\mv{D}_{ii}} \left( \mv{X}_i  - \mv{T}_{i,i+1:N}^{\trsp} \mv{X}_{i+1:N}  \right)\right] < \V[  \mv{X}_j  - \mv{T}_{j,j+1:N}^{trsp} \mv{X}_{1:j-1} ]$}.
$$
However  $ \mv{T}_{j,j+1:N} =  \mv{\Sigma}^{-1}_{j+1:N,j+1:N} \mv{\Sigma}_{j,j+1:N}$ which is the best linear predictor \citep[see][section 1.2]{stein1999interpolation}, and thus minimize \eqref{eq:mse} which is a contradiction.
\end{proof}

In the case when $\mv{\Sigma}$ is unknown and one has $i$ samples of $\mv{X}$, one can get the least square (LS) estimate of the Cholesky factor, $\hat{\mv{L}}$, by using the LS solution of \eqref{eq:mse} for each $j$. This is equivalent to replacing $\mv{\Sigma}$ in the theorem above with an empirical estimate $\hat{\mv{\Sigma}}$. Using the Delta method, one can show that $\hat{\mv{L}}$ converges to $\mv{L}$ as $i \rightarrow \infty$ \citep[e.g.][Theorem 20.8]{van2000asymptotic}.

In general, $\hat{\mv{L}}$ is the same matrix as the Cholesky factor of the inverse of the empirical covariance matrix. And in fact, computing $\hat{\mv{L}}$ by first computing $\hat{\mv{\Sigma}}^{-1}$ and then the cholesky factor require less computational effort than the regression method presented above. However, if $\pi$ has a sparse partial correlation structure, this can be utilized in the regression method to more efficiently construct $\hat{\mv{L}}$. Suppose, for example, that $\mv{X}_j$ and $\mv{X}_{j+2:n}$ are uncorrelated given $\mv{X}_{j+1}$. Then it follows that $\mv{T}_{j,j+2:n}=0$, and thus to compute $\mv{T}_{j,.}$ only the $1\times 1$ matrix $\mv{\Sigma}_{j+1,j+1}$ needs to be inverted, instead of the $N-j+1 \times N- j+1$  matrix $\mv{\Sigma}_{j+1:N,j+1:N}$.

Using these ideas we construct Algorithm \ref{alg:Lupdate}, that online updates $\mv{L}^{(i)}$ to $\mv{L}^{(i+1)}$ given a new observation $\mv{X}$. The list $A_j$ used in the algorithm is a set of indices such that the variable $\mv{X}_{j}$ is assumed to be partially uncorrelated with $\mv{X}_{\{j+1:N\} \setminus A_{j}}$ given $\mv{X}_{A_{j}}$. In line 6 of the algorithm, Sherman-Morisson formula \citep{bartlett1951} is used to update the lower triangular part of the inverse of $\mv{\Sigma}_{A_j,A_j}$. This reduces the complexity of each iteration in the for-loop from $\O(|A_j|^3)$ to  $\O(|A_j|^2)$. Thus, the algorithm has a computational complexity bounded by $\O(N \max_j |A_j|^2)$. This complexity should be compared to updating  $\mv{\Sigma}$ (or its Cholesky factor) which has a complexity of $\O(N^2)$.

 \begin{algorithm}[t]
 \caption{Cholesky factor updating}\label{alg:Lupdate}
 \begin{algorithmic}[1]
 \Procedure{L-update}{$\mv{X},\{\hat{\mv{\Sigma}}^{-1}_{A_j,A_j}, \hat{\mv{\Sigma}}_{A_j\cup j,j},A_j\}_{j=1}^N,i$}

\State $\mv{D} \gets \mv{0}_{N\times N}$.
\State $\mv{T} \gets diag(\mv{1}_{1 \times N})$
    \For{$j=1,\dots,N$}
    	\State $ \hat{\mv{\Sigma}}_{(A_j \cup j),j} \gets  \frac{i-1}{i}\hat{\mv{\Sigma}}_{(A_j \cup j),j} + \frac{\mv{X}^{rsp}_{(A_j \cup j)}\mv{X}_j}{i}$
    	\State  $\hat{\mv{\Sigma}}_{A_j,A_j}^{-1} \gets SM( \frac{i-1}{i}\hat{\mv{\Sigma}}_{A_j,A_j}^{-1}, \frac{\mv{X}_{A_j}}{i}) $ \Comment{Updating the matrix using Sherman-Morrsion}
    	\State $\mv{T}_{j,A_j}  \gets \hat{\mv{\Sigma}}_{A_j,A_j}^{-1}  \hat{\mv{\Sigma}}_{A_j,j} $
    	\State $ \mv{D}_{j,j}   \gets   \hat{\mv{\Sigma}}_{j,j} -  \hat{\mv{\Sigma}}^{\trsp}_{A_j,j} \mv{T}_{j,A_j} $
    \EndFor
    \State $\mv{L}  \gets \mv{T}^{\trsp}\mv{D}^{{-1/2}} $
    \State \textbf{return} $\{  \mv{L},\hat{\mv{\Sigma}}^{-1}_{A_j,A_j}, \hat{\mv{\Sigma}}_{(A_j \cup j),j} , \}_{j=1}^N$
 \EndProcedure
 \end{algorithmic}
 \end{algorithm}

Note that the algorithm is simple to parallelize since the iterations in the loop can be executed independently of each other.

\subsection{Finding the conditional sparsity structure}\label{sec:SparP}
So far it has been assumed that we know the partial correlation structure of the target distribution, which is needed to construct the set $\{A_i\}_{i=1}^N$ in Algorithm \ref{alg:Lupdate}. However, it is often not practically feasible to derive the partial correlation structure, which is problematic if we want to use Algorithm \ref{alg:Lupdate}. Because of this, we will in this section present an algorithm that estimates the set $\{A_i\}_{i=1}^N$.

Instead of estimating the partial correlation structure directly, we estimate the conditional dependence structure and then use this to construct a sparity pattern $\{A_i\}_{i=1}^N$. Although conditional independence does not imply that the variables are partially uncorrelated \citep[see][for a counter example]{wermuth1998statistical} its seems reasonable to use the conditional dependece structure if no other knowledge is available. Here it is also important to stress that the adaptation algorithm is valid even without using the correct parital correlation structure. In fact, one can use any sets $\{A_i\}_{i=1}^N$ in the algorithm and it for instance estimates, in a somewhat complicated way, a diagonal covariance matrix if $A_i=\{\emptyset\}$. Of course, in order to get good mixing one should use a dependency structure that contains the elements that are strongly correlated in the target density.

Below, we show how to estimate the dependence structure and then how to create  $\{A_i\}_{i=1}^N$ given the structure. However, we first need to define more precisely what we are estimating. To that end, we define variables $\mv{X}_i$ and $\mv{X}_j$ to be conditionally independent if $\,$ $\pi( \mv{X}_i| \mv{X}_{- \{i,j\}}, \mv{X}_j  ) = \pi( \mv{X}_i| \mv{X}_{- \{i,j\}} ) $. The conditional dependency structure of $\pi$ can be represented using an undirected graph $\mathcal{G}=(\{1,\ldots,N\},E)$, where $\{i,j\}$ is in the set of edges, $E$, if and only if $\mv{X}_i$ and $\mv{X}_j$ are conditionally dependent, see \cite{rue2005gaussian} for more details. Thus the goal first goal is to find $E$.

Let $\mv{e}_i$ represent the canonical basis vector, that is $e_{i,j} = 1$ for $j=i$ and $e_{i,j} = 0$ otherwise. The main idea is that one can estimate $E$ by
\begin{equation}\label{eq:Ehat}
\scalebox{0.9}{$\hat{E}=  \bigcup_{i=1}^N \left\{ \{i,j \} : \frac{\partial \log  \pi(\mv{x})}{\partial \mv{x}_j} -   \frac{\partial \log  \pi(\mv{x}+\mv{e}_i)}{\partial \mv{x}_j} \neq  0\right\} \setminus \{i,i \}$}.
\end{equation}

The following proposition proves that $\hat{E}\subseteq E$. That is, the estimated conditional dependency structure always has at most as many edges as the true dependency structure, but is often identical to $E$.

\begin{prop}\label{prop:spartpattern}
Let $\mv{X}$ have density $\pi$ and define $\hat{E}$ through \eqref{eq:Ehat}. If $\{i,j\}\in \hat{E}$ then $\mv{X}_{j}$ is not conditionally independent of $\mv{X}_i$ given $\mv{X}_{-\{i,j\}}$.
\end{prop}

\begin{proof}
Assume that $\{i,j\}\in E \cap \hat{E}$ and that $\mv{X}_j$ and $\mv{X}_i$ are conditionally independent. By the conditional independence
\begin{align*}
\log \pi(\mv{x}) &= \log  \pi(\mv{x}_j| \mv{x}_{-j}) + \log \pi(\mv{x}_{-j})  \\
&=\log  \pi(\mv{x}_j| \mv{x}_{- \{j,i\}}) + \log  \pi(\mv{x}_{-j}),
\end{align*}
and thus
\begin{align*}
\frac{\partial \log  \pi(\mv{x})}{\partial \mv{x}_j} = \frac{\partial \log  \pi(\mv{x}+\mv{e}_i)}{\partial \mv{x}_j},
\end{align*}
which contradicts $j \in  \hat{E}$.
\end{proof}

To find $\{A_i\}_{i=1}^N$ given $\hat{E}$, note that $\hat{E}$ is equivalent to a zero pattern of a symmetric matrix $\mv{\mathcal{Q}}$. From the non-zero pattern of $\mv{\mathcal{Q}}$ one can construct a non-zero pattern of its Cholesky factor.  The sparsity pattern of the Cholesky factor is not invariant to the ordering of the vector $\mv{X}$. For a sparse symmetric matrix one can in general reduce the number of non-zero elements in its Cholesky factor by using a clever ordering, and because of this we add a reodering step to the algorithm. There are several possible reordering methods one could use in this step, and we use an Approximate Minimum Degree (AMD) \citep{davis2006direct} method. Figure \ref{fig:reodering} displays the effect of using an AMD reordering for an example presented in the next section. The reodering reduced the number of symbolic non-zero elements in  $\mv{L}$ from $9689$ to $1380$.

Putting all steps together results in Algorithm \ref{alg:Aupdate}. It should be noted that the algorithm can be modified so that $\{A_j\}^N_{j=1}$  is updated online within an MCMC iteration. Also note that numerical round-of errors can cause the statement of Proposition \ref{prop:spartpattern} to fail when implementing the algorithm numerically, however we have not experienced any problems with this on the applications we have tested the method on so far.

\begin{algorithm}[h]
 \caption{Conditional dependence estimation}
 \label{alg:Aupdate}
 \begin{algorithmic}[1]
 \Procedure{A-update}{$\mv{X},\nabla \log\pi $}
 \State $E=\emptyset$
  \For{$i=1,\dots,N$}
    \State $\mv{dX} \gets \nabla \log\pi(\mv{X}) $
     \State $\mv{dX}^e \gets \nabla \log \pi(\mv{X} +  \mv{e}_i) $  \Comment{$\mv{e}$ is a canonical basis vector}
      \State  $E \gets E \cup \{\{i,j\} : \mv{dX}^e_j - \mv{dX}_j \neq 0    \}$
    \EndFor
    \State  $\mv{\mathcal{Q}} \gets \mbox{symbolicMatrix}( (\{1,\ldots,N\},E)$ \Comment{the matrix sparsity pattern}
    \State $\mv{r} \gets AMD(\mv{\mathcal{Q}} )$								\Comment{$\mv{r}$ is a reordering}
      \State $\mv{\mathcal{L}} \gets \mbox{Cholesky}(\mv{\mathcal{Q}} _{\mv{r}\mv{r}})$    \Comment{the resulting sparsity pattern}
      \For{$i=1,\dots,N$}
        \State  $A_i \gets  \{j: \mv{\mathcal{L}}_{i,\cdot} \neq 0   \} \setminus \{i\}$
         \EndFor
    \State \textbf{return} $\{A_i\}^N_{i=1},\mv{r}$
 \EndProcedure
 \end{algorithmic}
 \end{algorithm}

 \begin{figure}[t]
 \centering
 \includegraphics[width = 0.3\linewidth]{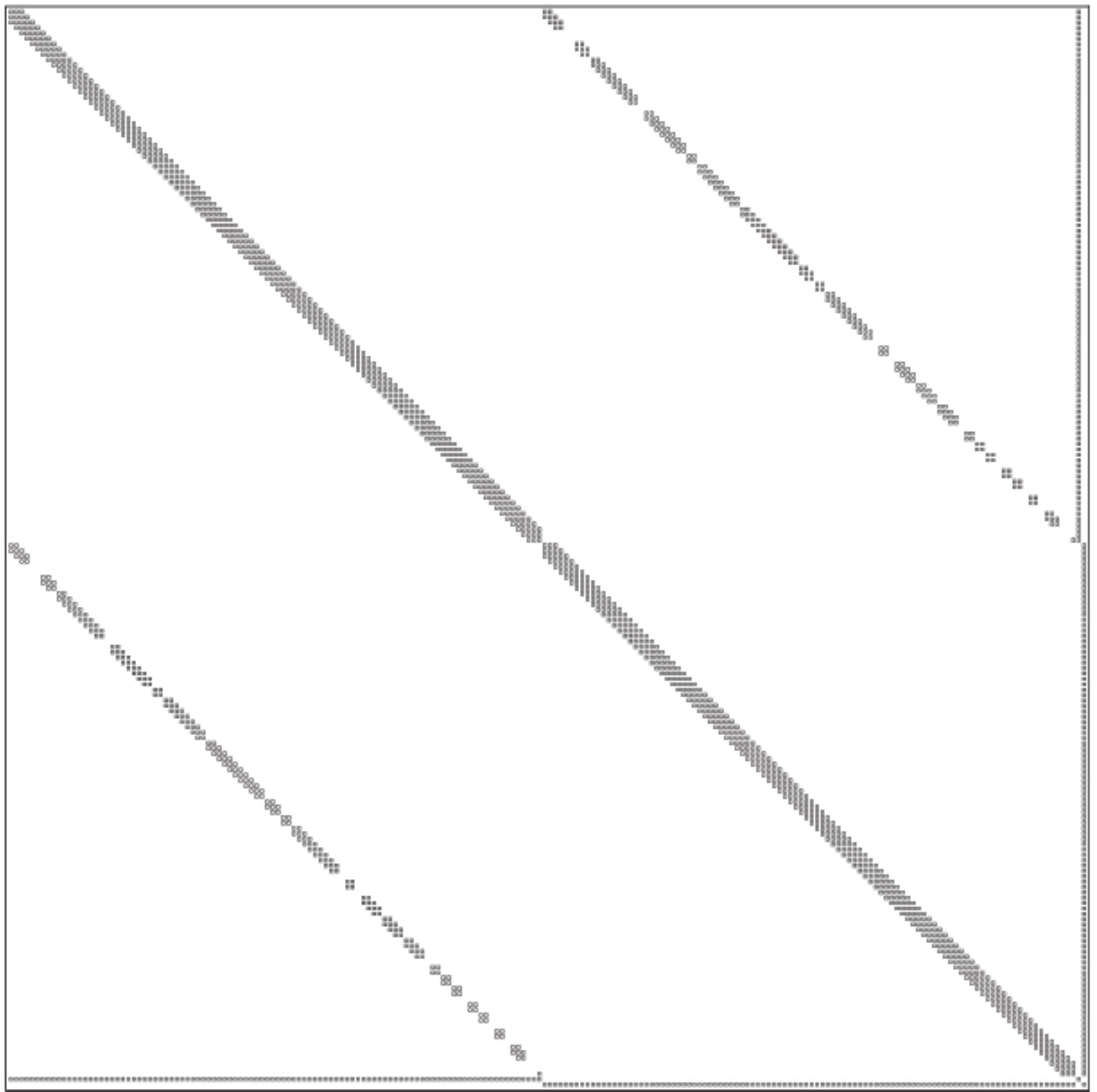}
 \includegraphics[width = 0.3\linewidth]{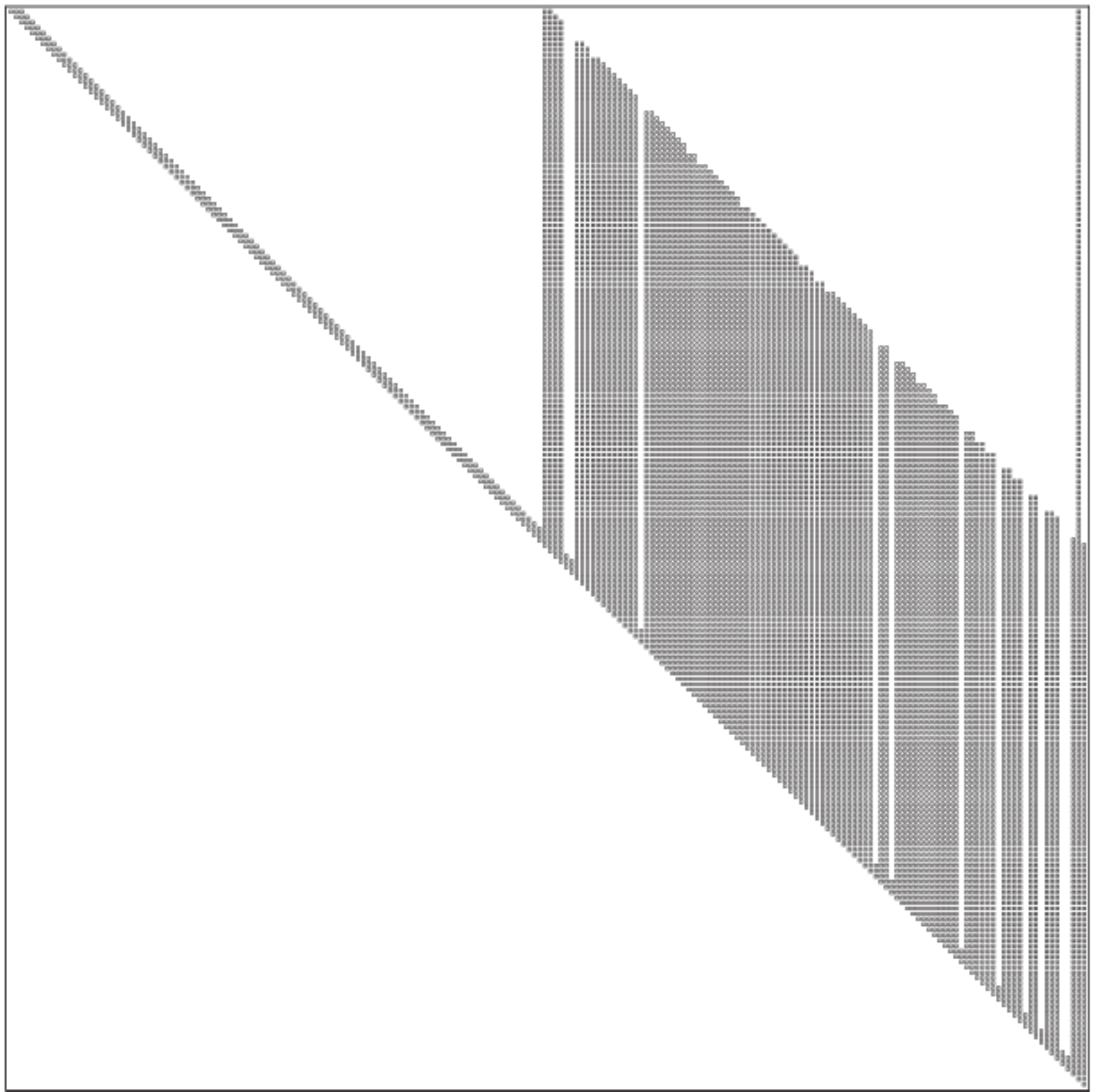}
 \includegraphics[width = 0.3\linewidth]{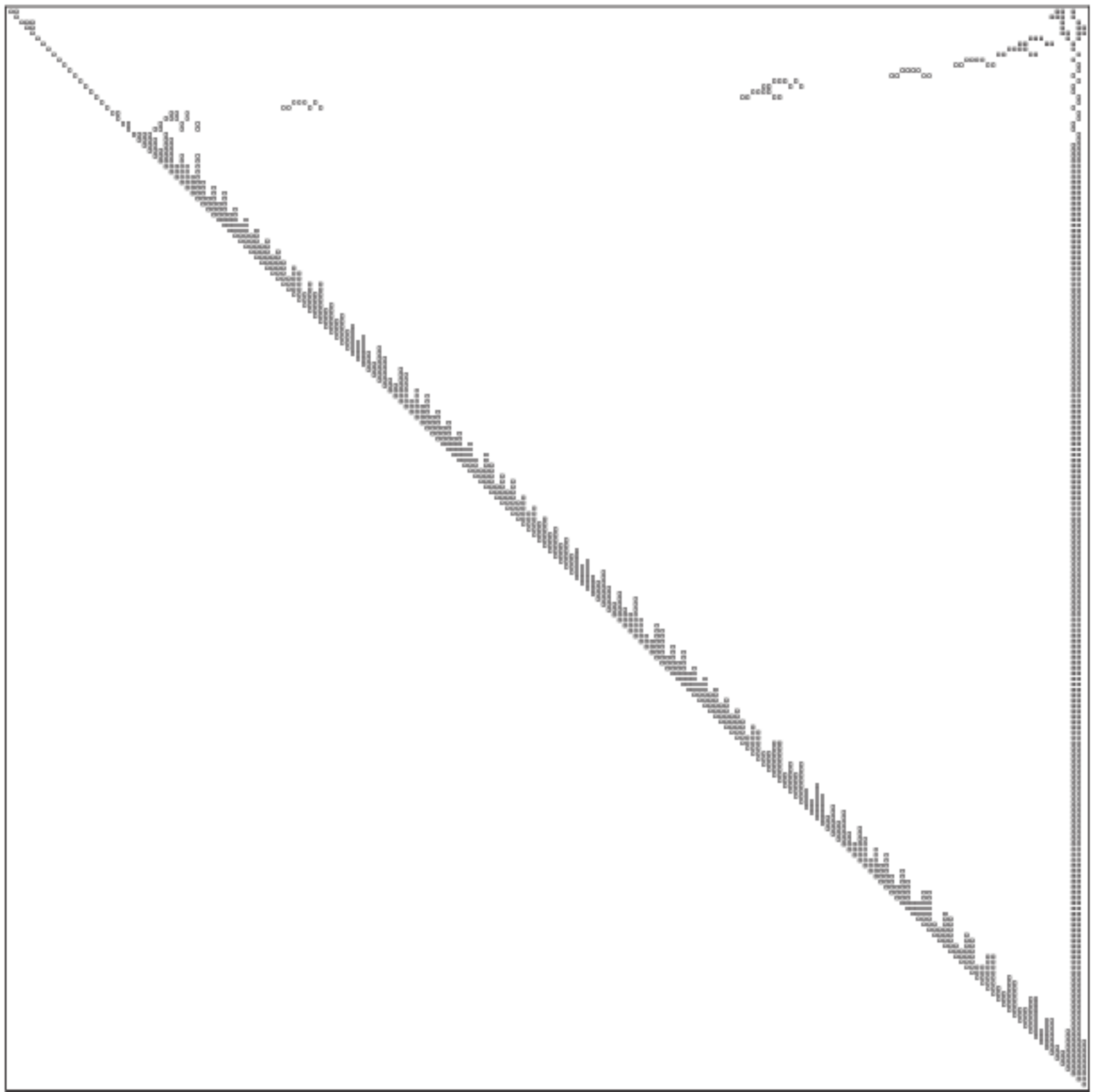}
 \caption{Conditional dependency structure for the spline model used in the example in Section \ref{sec:spline}. The sparsity pattern $\mv{\mathcal{Q}}$ (left), the sparsity pattern $\mv{\mathcal{L}}$ (middle), and the sparsity pattern $\mv{\mathcal{L}}$ after reordering the nodes using an AMD reodering (right). The number of non-zero elements in the matrices are $1860$ (left), $9689$ (middle), and $1380$ (right) and the size of the matrices are $200\times 200$. }
 \label{fig:reodering}
 \end{figure}

\section{Examples}
In this section, we compare the performance of the precision-based adaption method to the covariance adaptation method using two examples.

The purpose of the first example is two-folded: First, to show that the precision adaption converges to the correct scaling matrix and to show the rate of convergence. We use a Gaussian density because this gives an explicit value of (mixing) performance of a scaling matrix compared to the optimal mixing. The second purpose is to study the behavior of adaption under an increasing dimension of $\pi$. For this example we display only the results for the MALA since the results for the MHRW tells the same story.

For the second example, we have chosen a density with very sparse conditional dependence, but with strong correlation between the parameters. Here we compare adaption for both MHRW and MALA, and we also estimate the sparsity using the the conditional dependence since we do not know the partial correlation structure.

 In both examples we additionally adapt a scaling factor of the proposal, using the method proposed by \cite{EAMCMC_Roberts}, so that the acceptance rate is $0.234$ for MHRW and $0.574$ for MALA.

\subsection{An example from spatial statistics}
This first example is a latent Gaussian model taken from spatial statistics. Let $X({s})$, ${s}\in [0,1] \times [0,1]$, be a mean-zero Gaussian Mat\'ern fields represented as a solution to a stochastic partial differential equation
 $(\kappa^2-\Delta)^{\frac{\alpha}{2}} X(\mv{s}) =\mathcal{W}(\mv{s})$, where $\alpha$ is a shape parameter, $\kappa^2$ is the parameter controlling the correlation range of $X$, $\mathcal{W}$ is Gaussian white noise, and $\Delta$ is the laplacian. The field is discretized using a finite element method (FEM), resulting in an approximation
 \begin{equation}\label{eq:basis}
 X(\mv{s}) = \sum_{i=1}^n \varphi_i(\mv{s})x_i
 \end{equation}
 where $\{\varphi_i\}$ are piecewise linear basis functions induced by a triangulation of the domain and $\mv{x} = (x_1, \ldots, x_n)$ is a mean-zero Gaussian Markov random field with sparse precision matrix $\mv{Q}$. See \cite{lindgren2011explicit} for further details of the construction.

 The field is observed under Gaussian measurement noise at $100$ locations $\mv{s}_1, \ldots, \mv{s}_{100}$ chosen at random in the domain, resulting in observations $y_i = X(s_i) + \varepsilon_i$, where $\varepsilon_i$ are iid $\pN(0,\sigma^2)$ variables. The resulting posterior distribution for $\mv{x}$ given the observations is
 \begin{equation}\label{eq:spde2}
\scalebox{0.9}{$  \pi(\mv{x}|\mv{y})  \propto  \exp \left( - \frac{1}{\sigma^2} \left( \mv{y} - \mv{Ax} \right)^{\trsp} \left( \mv{y} - \mv{Ax} \right) - \frac{1}{2} \mv{x}^{\trsp} \mv{Q}^{-1} \mv{x} \right)$},
 \end{equation}
 where the matrix $\mv{A}$ is a sparse observation matrix that links the observation locations to the random weights $\mv{x}$ through $A_{ij} = \varphi_j(\mv{s}_i)$.

 The number of basis functions, $n$, controls the accuracy of the FEM approximation, and by increasing it, we can study how well the adaptation method scales. \cite{cotter2013} studied MCMC methods for discretizations of continuous functions and pointed out that many MCMC methods scales poorly with an increasing number of basis function. It is therefore interesting to compare precision adaption with regular covariance adaptation for this example when we increase $n$.

 Since both the prior and the likelihood are Gaussian, the posterior density of $\mv{x}|\mv{y}$ can be shown to be Gaussian with covariance matrix $ \mv{\Sigma} = \left( \mv{Q} +\frac{1}{\sigma^2}\mv{A}^{\trsp}\mv{A}  \right)^{-1}$. Thus, this covariance matrix is the optimal scaling matrix for the MALA. As a measure of how good another proposal matrix $\mv{\Sigma}_{p}$ is, we use
\begin{equation}\label{eq:b}
 b = n \frac{\sum_{i=1}^n \lambda_i}{(\sum_{i=1}^n \sqrt{\lambda_i})^2},
\end{equation}
 where $ \{\lambda_i \}$ are the eigenvalues of $\mv{\Sigma} \mv{\Sigma}_p^{-1}$. The optimal value of $b$ is one and the larger value the worse proposal matrix, see \cite{roberts2001} for details of this measure.

 Figure \ref{fig:GaussianAdaptation} displays the convergence of $b$ in AMCMC simulations of \eqref{eq:spde2} with $n=10^2, 20^2, 30^2,$ and $40^2$ basis functions in the FEM discretization. Although $40^2$ is a rather small dimension for a Gaussian random field, it is quite large for a black box adaptation method. In the figure, one can clearly see that the precision adaptation converges, at least initially, much faster towards a reasonable proposal matrix compared with the covariance adaptation method. Further, from Figure \ref{fig:GaussianTime} one sees that both sampling the MALA and computing the Cholesky factor the precision adaption scales much better with size of the dimension compared to the regular covariance adaption.

 In conclusion, we see that the precision adaption outperforms the regular covariance adaption by a large margin, although neither method is great without a good initial guess of the covariance matrix.

 \begin{figure}[ht]
 \begin{center}
 \includegraphics[width=0.35\linewidth]{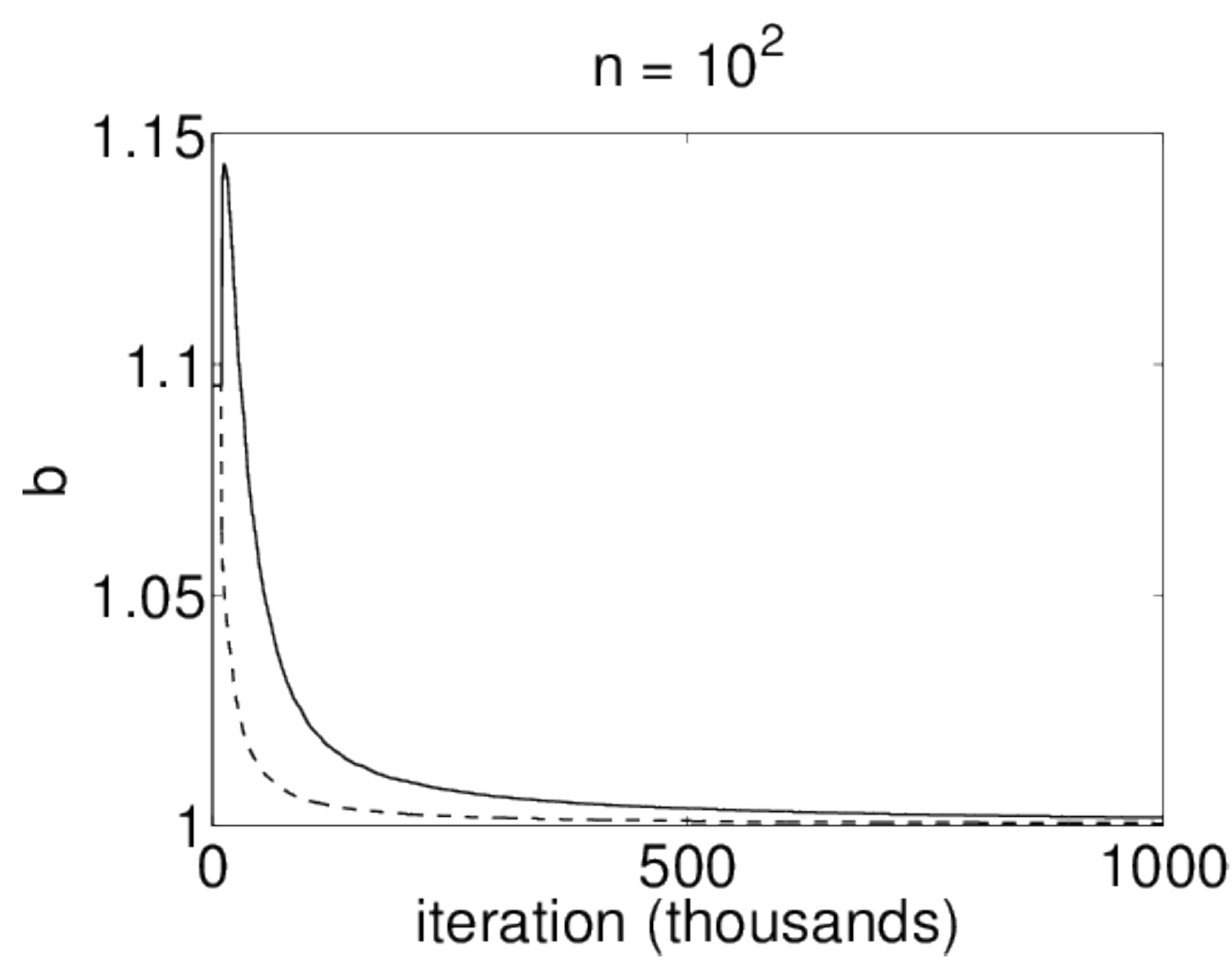}
 \includegraphics[width=0.35\linewidth]{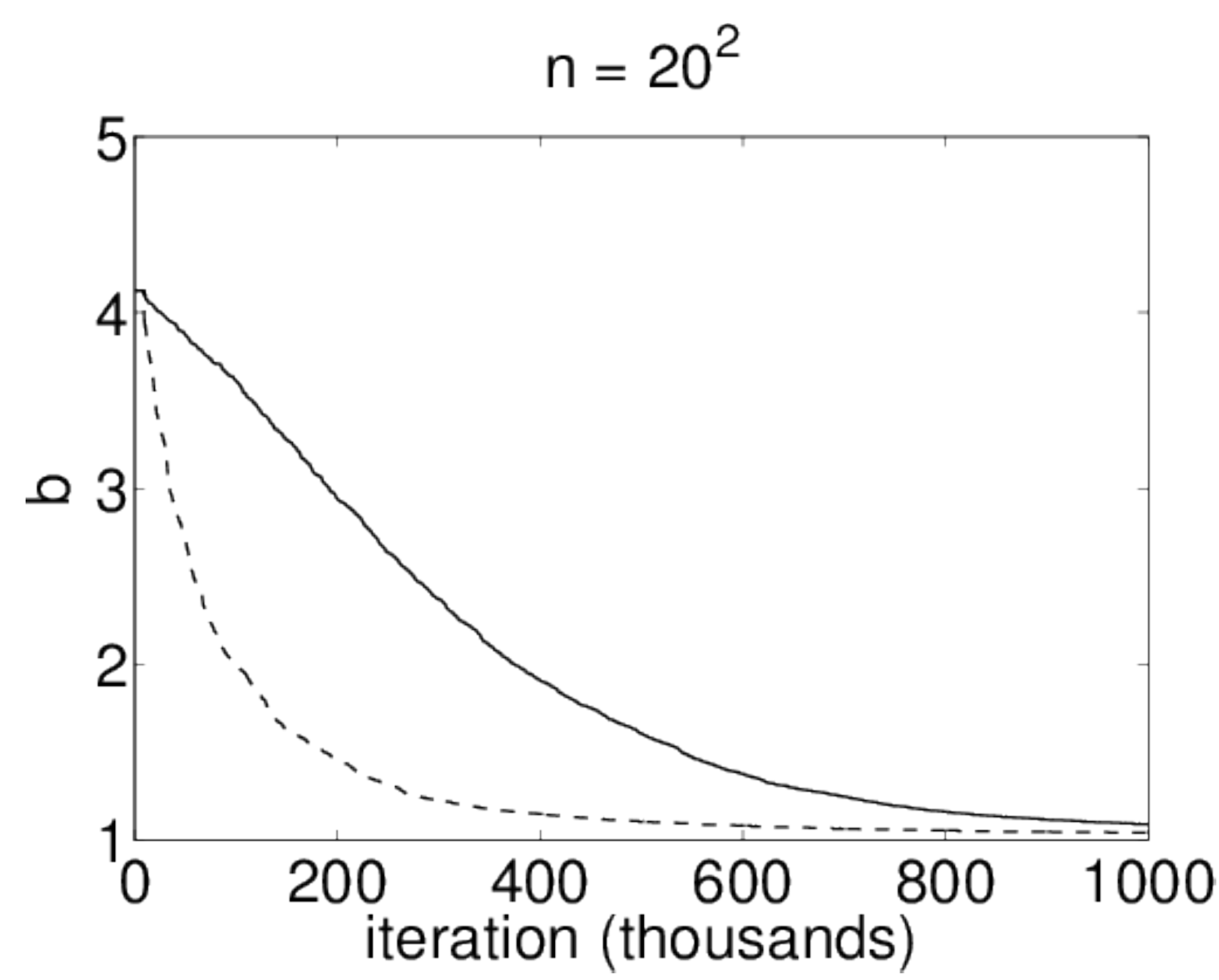}\\
 \includegraphics[width=0.35\linewidth]{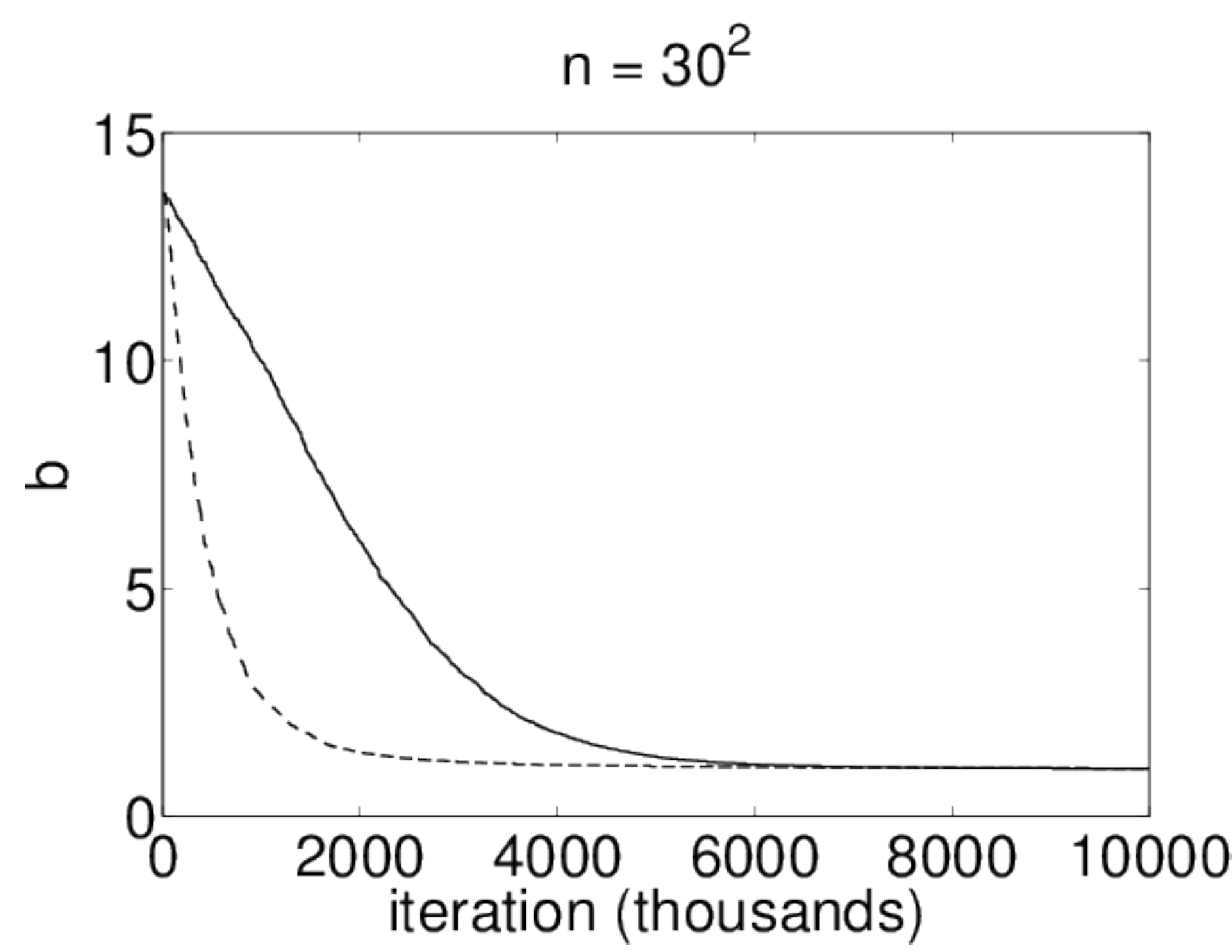}
 \includegraphics[width=0.35\linewidth]{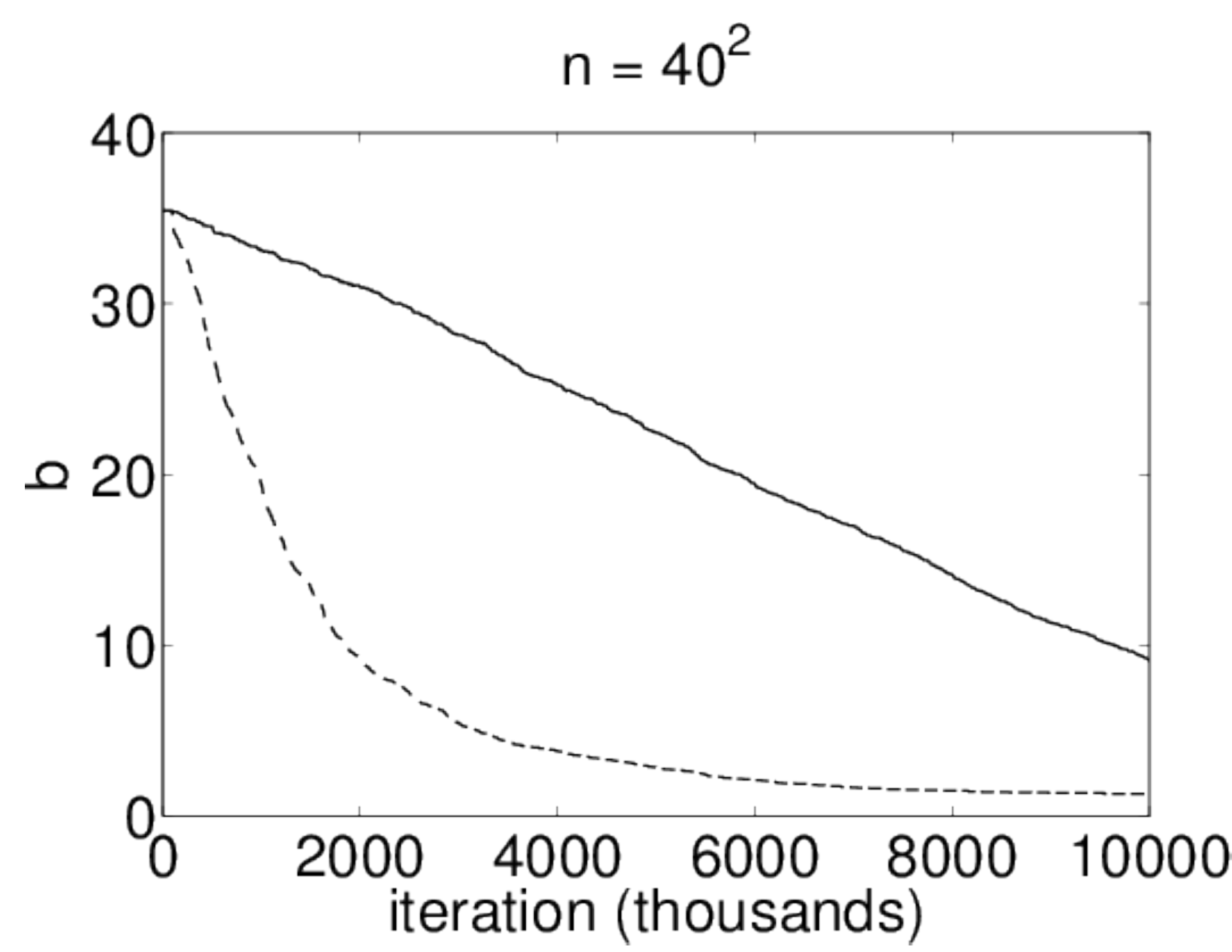}
 \end{center}
 \caption{The factor $b$ defined in equation \ref{eq:b}, which is a measure of how good the proposal is, for various discretization sizes. The optimal value of $b$ is $1$. The dashed line is the precision adaptation and the solid line is regular covariance adaptation. Note that the number of iterations varies between the plots at the bottom panels and the top panels.}
 \label{fig:GaussianAdaptation}
 \end{figure}

 \begin{figure}[ht]
 \begin{center}
 \includegraphics[scale=0.5]{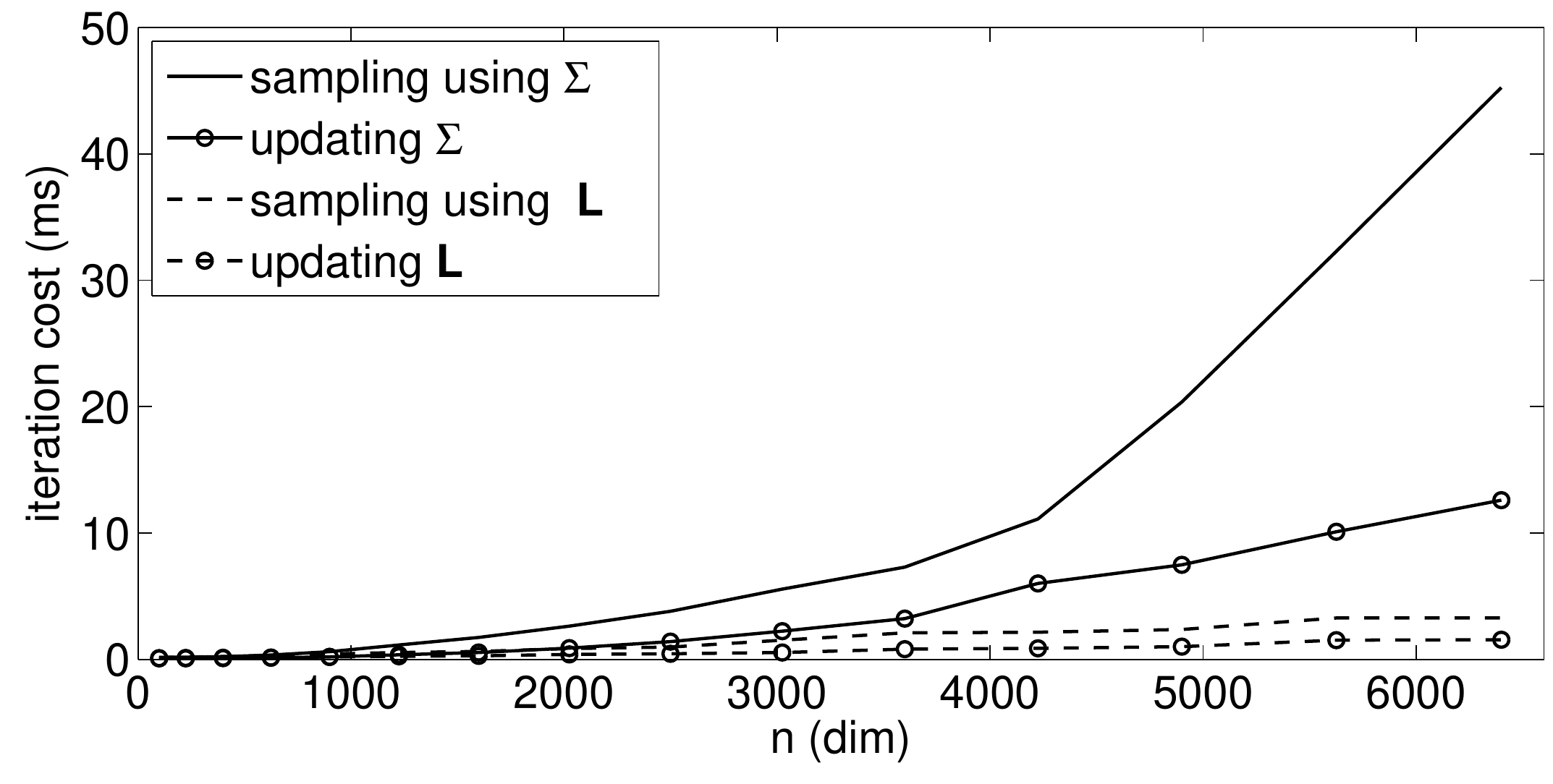}
 \end{center}
 \caption{The solid line shows average sampling time (in ms) of using covariance based scaling matrix as a function of the dimension of $\mv{x}$. The  dashed line shows the average sampling time using precision based scaling matrix. The  solid circled line shows the average time of updating the Cholesky factor of the covariance matrix. The dashed circled line shows the average time of updating the Cholesky factor of the precision matrix.}
 \label{fig:GaussianTime}
 \end{figure}

 \subsection{ Adaptive spline smoothing}\label{sec:spline}
 In this example we compare the adaption methods on a Bayesian model with a sparse conditional dependency structure. The example is an adaptive smoothing spline model with varying measurement error, and stochastic differential equations are again used to define the model. Let $X(t)$ be a twice differentiable second order random walk, defined as the solution to
 $$
 \tau_X \Delta X(t) =\mathcal{W}_1(t),
 $$
 where $\mathcal{W}_1$ is Gaussian white noise. As in the previous example, we measure $X$ under Gaussian noise, resulting in observations $Y(t_i) =  X(t_i) + \varepsilon_i$. The difference here is that we assume that the measurement noise $\varepsilon$ has a varying variance, $V(\varepsilon_i) = e^{2V(t_i)}$, where $V(t_i)$ is a second order random walk:
 $$
 \tau_V \Delta V(t) =  \mathcal{W}_2(t).
 $$
 Here $W_2(t)$ again is Gaussian white noise. As in the previous example, the differential equations are discretized using a FEM approach, described in detail by \cite{lindgren2008second}. The resulting joint posterior distribution is
 \begin{align}\label{eq:spde3}
\log \pi(\mv{x},\mv{v},\sigma_V^2,\sigma_X^2|\mv{y}) \propto& - \frac{1}{2}\left( \mv{y} - \mv{A}_x \mv{x} \right)^{\trsp} \mv{D}(e^{-2\mv{A}_v \mv{v}})  \left( \mv{y} -  \mv{A}_x \mv{x} \right) \nonumber\\
 &-\sum_{i=1}^{n} v_i - \frac{\tau_x}{2} \mv{x}^{\trsp} \mv{Q}_x \mv{x} -  \frac{\tau_v}{2} \mv{v}^{\trsp} \mv{Q}_v \mv{v}  \nonumber \\
 & +  \frac{n}{2} \log(\tau_v)  +\frac{n}{2} \log(\tau_x) - \tau_v  - \tau_x.
 \end{align}
 Here $\mv{D}(\mv{v})$ is a diagonal matrix with $\mv{v}$ on the diagonal. The matrices $\mv{A}_x$ and $\mv{A}_v$ are as in the previous example observation matrices that link the observation locations to the random weights $\mv{x}$ and $\mv{v}$ for the FEM discretizations. The matrices $\mv{Q}_x$ and $\mv{Q}_v$ are sparse tridiagonal matrices, given in explicit form in \citep{lindgren2008second}. Finally, the last terms in \eqref{eq:spde3} come from assuming exponential priors on the precision parameters.

 We apply the model to a classical data set of motor-cycle crashes, analysed by \cite{Silverman_spline}. The observations are accelerometer readings taken through time in simulated crashes used to test crash helmets. Figure \ref{fig:data2} displays the data points and it is clear that the variance is not constant over time.

 We use $250$ piecewise linear basis function for both $\mv{x}$ and $\mv{v}$ in the basis expansion \eqref{eq:basis}. In Figure  \ref{fig:data2}, the solid line is the posterior mean of the spline function $X(t)$ and the dashed lines show the posterior means of $X(t) \pm 1.96e^{-V(t)}$.

 To compare the adaption for MHRW and MALA, we study the convergence of the samples for $x_{21}$ (which corresponds to $X(6.8)$) and $\tau_{v}$. One expects the parameter $\tau_v$ to be more difficult to sample than $x_{21}$ since it is higher up in the hierarchical structure.

 In Figure \ref{fig:trace_tau} we display trace plots of ten million samples of $\log(\tau_v)$ for all adaption methods. It is apparent from the figure that the adapataions for MHRW did not converge before the algorithm was stopped, and it is not clear if the MALA method with covariance adaptation converged. Finally, MALA using precision adaptation converges to a steady state after approximately $2.5$ million samples. The results for $x_{21}$ in Figure \ref{fig:trace_X} tells a similar story.

 Table \ref{tab:comptime} displays the average cost of one iteration for each of the methods. Here it is worth pointing out that precision adapted MALA has a lower computational cost than regular MHRW with covariance adaption.

 The numbers should be put in relation with the number of basis functions used in the FEM approximations. The relative improvement of precision adaption compared with covariance adaption would increase if one were to use more basis functions, and it would decrease if fewer basis function were used.

 In conclusion, from the above result it is clear that one should use MALA with precision adaptation for this example.

 \begin{table}[ht]
   \begin{center}
     \begin{tabular}{lcc}
        & MHRW & MALA \\
       \hline
       covariance adaption & 0.67 & 1.32\\
       precision adaption & 0.49 & 0.59 \\       \hline
     \end{tabular}
   \end{center}
       \caption{Average time per iteration in ms for the different sampling-adapation schemes for the second example.}
       \label{tab:comptime}
 \end{table}

 \begin{figure}
 \centering
 \includegraphics[scale=0.45]{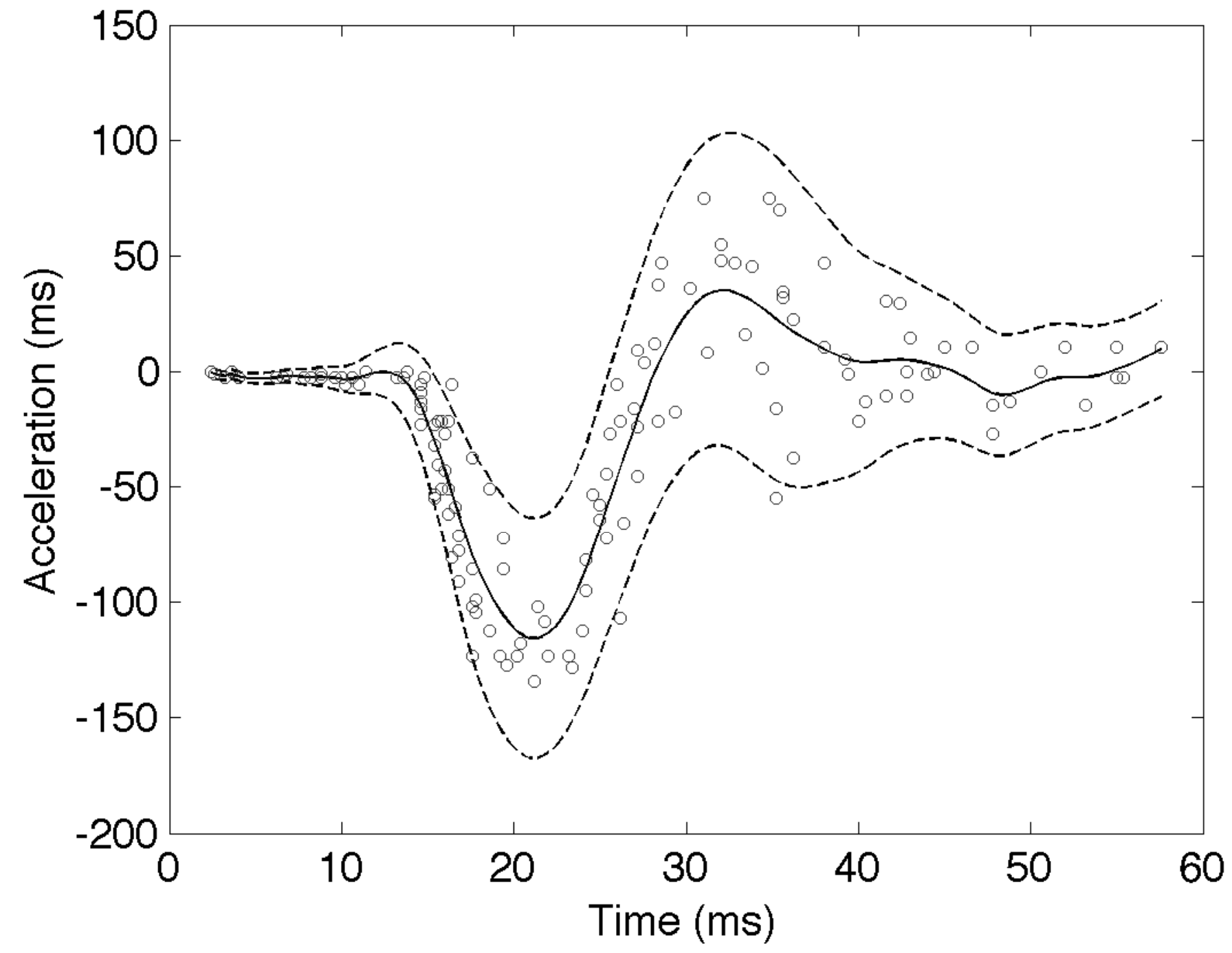}
 \caption{Data and results for example 2. The black dots are measurements, the solid line is the posterior mean of $X(t)$, and the dashed lines are the posterior mean of $ X(t) \pm 1.96 e^{-V(t)}$. }
 \label{fig:data2}
 \end{figure}

\begin{figure}[htb]
\begin{center}
 \begin{minipage}[b]{0.4\linewidth}
 \centering
 (a)
 \includegraphics[width=\linewidth]{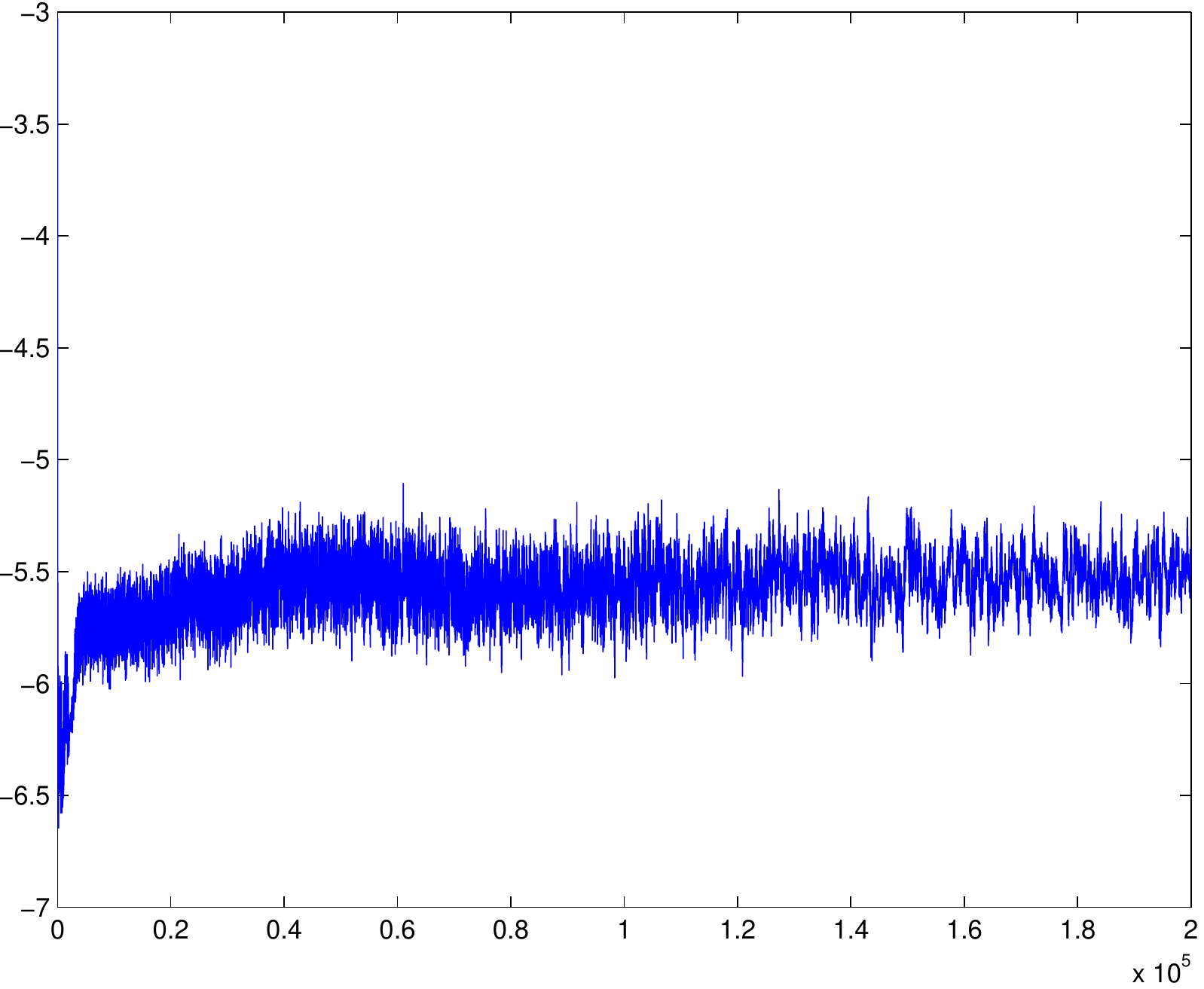}\\
 \end{minipage}
 \begin{minipage}[b]{0.4\linewidth}
 \centering
 (b)
 \includegraphics[width=\linewidth]{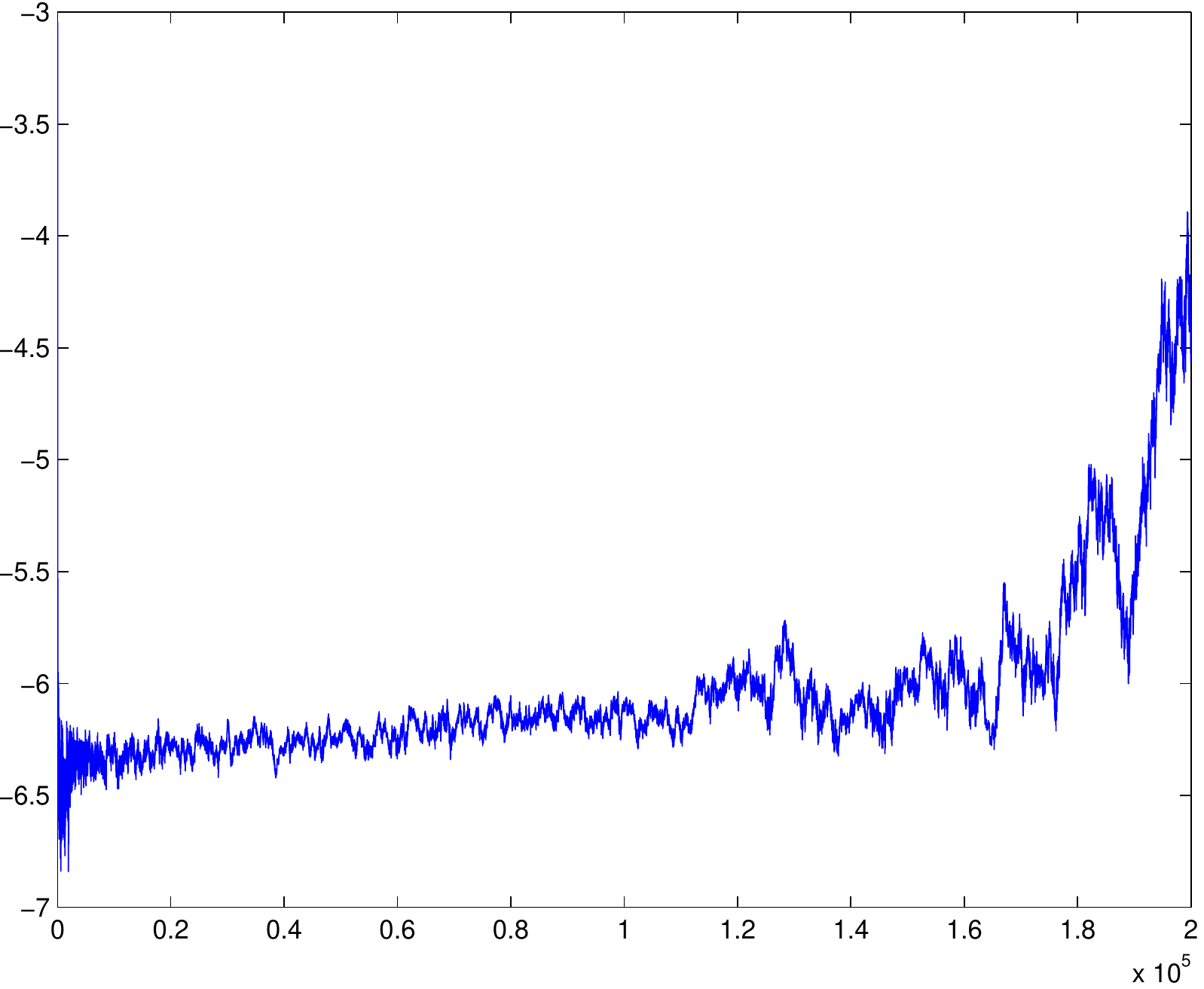}\\
 \end{minipage}\\
 \begin{minipage}[b]{0.4\linewidth}
 \centering
 (c)
 \includegraphics[width=\linewidth]{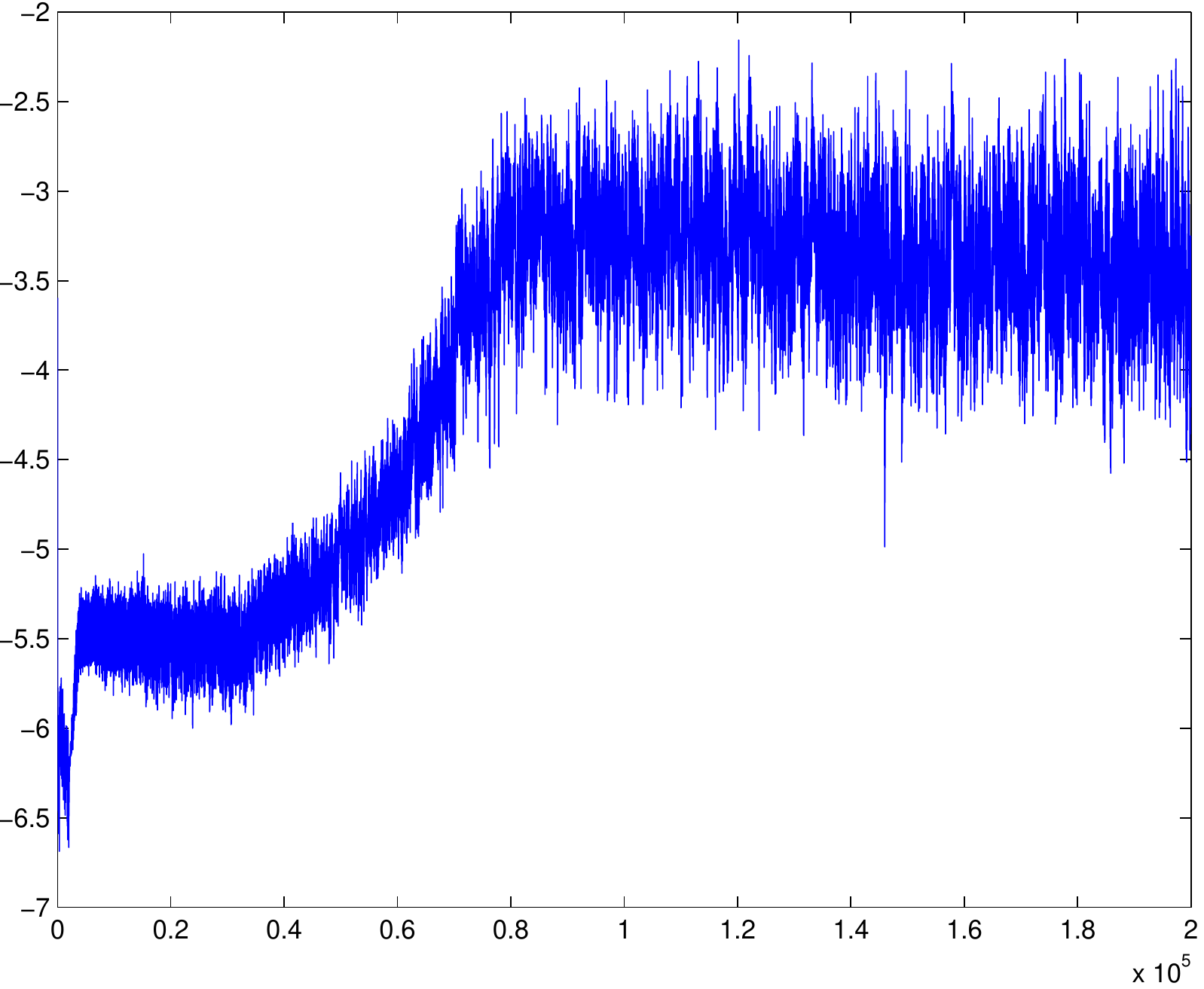}\\
 \end{minipage}
 \begin{minipage}[b]{0.4\linewidth}
 \centering
 (d)
 \includegraphics[width=\linewidth]{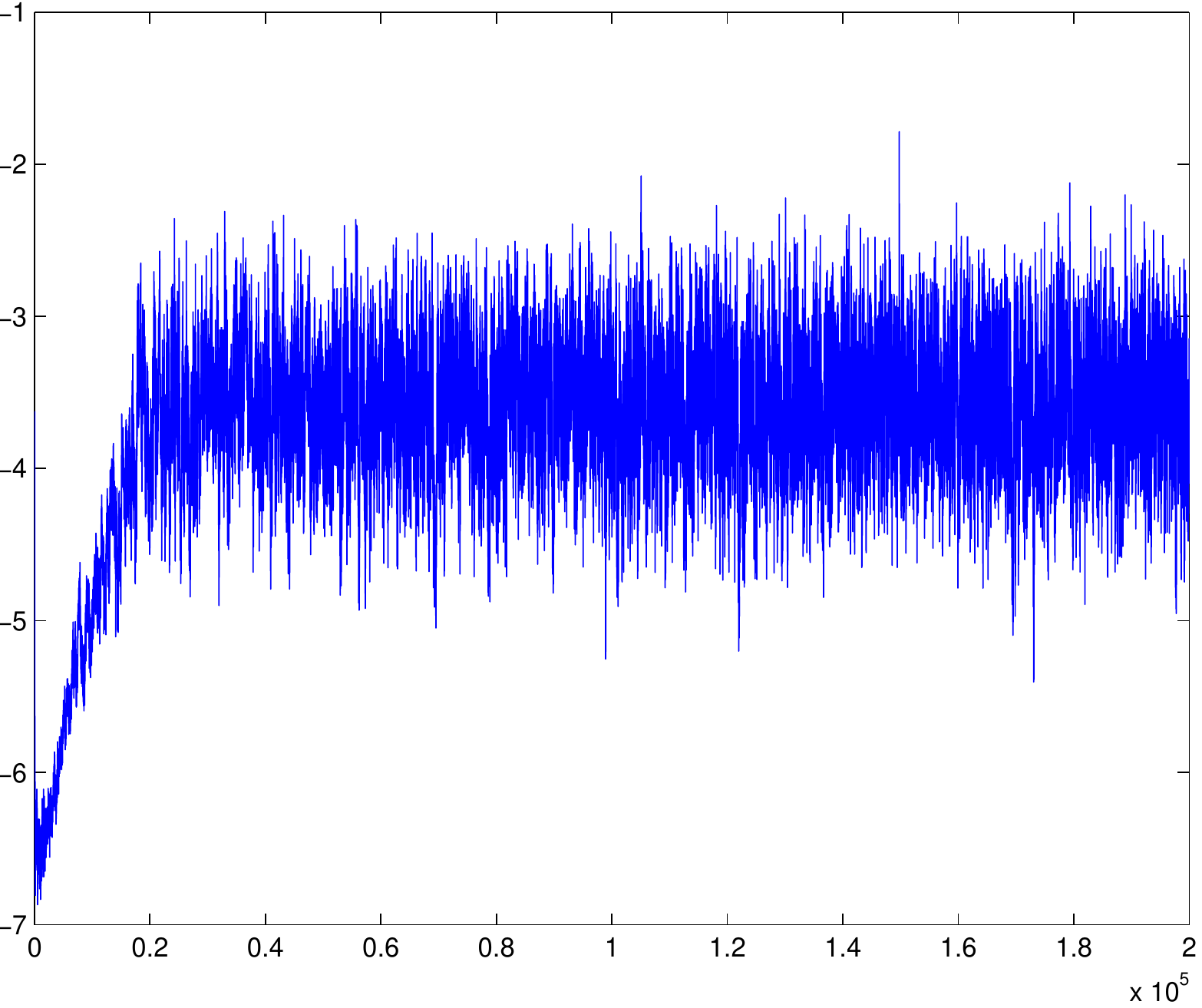}\\
 \end{minipage}
 \end{center}
        \caption{ The panels show trace plots of $\log(\tau_v)$. The samples are taken from runs of ten million MCMC iteration where every fiftieth sample is stored. The methods shown are MHRW with covariance adaptation (Panel a), MHRW with precision adaption (Panel b), MALA with covariance adaption (Panel c), and  MALA with precision adaption (Panel d). }
                \label{fig:trace_tau}
 \end{figure}

 \begin{figure}[htb]
\begin{center}
 \begin{minipage}[b]{0.45\linewidth}
 \centering
 (a)
 \includegraphics[width=\linewidth]{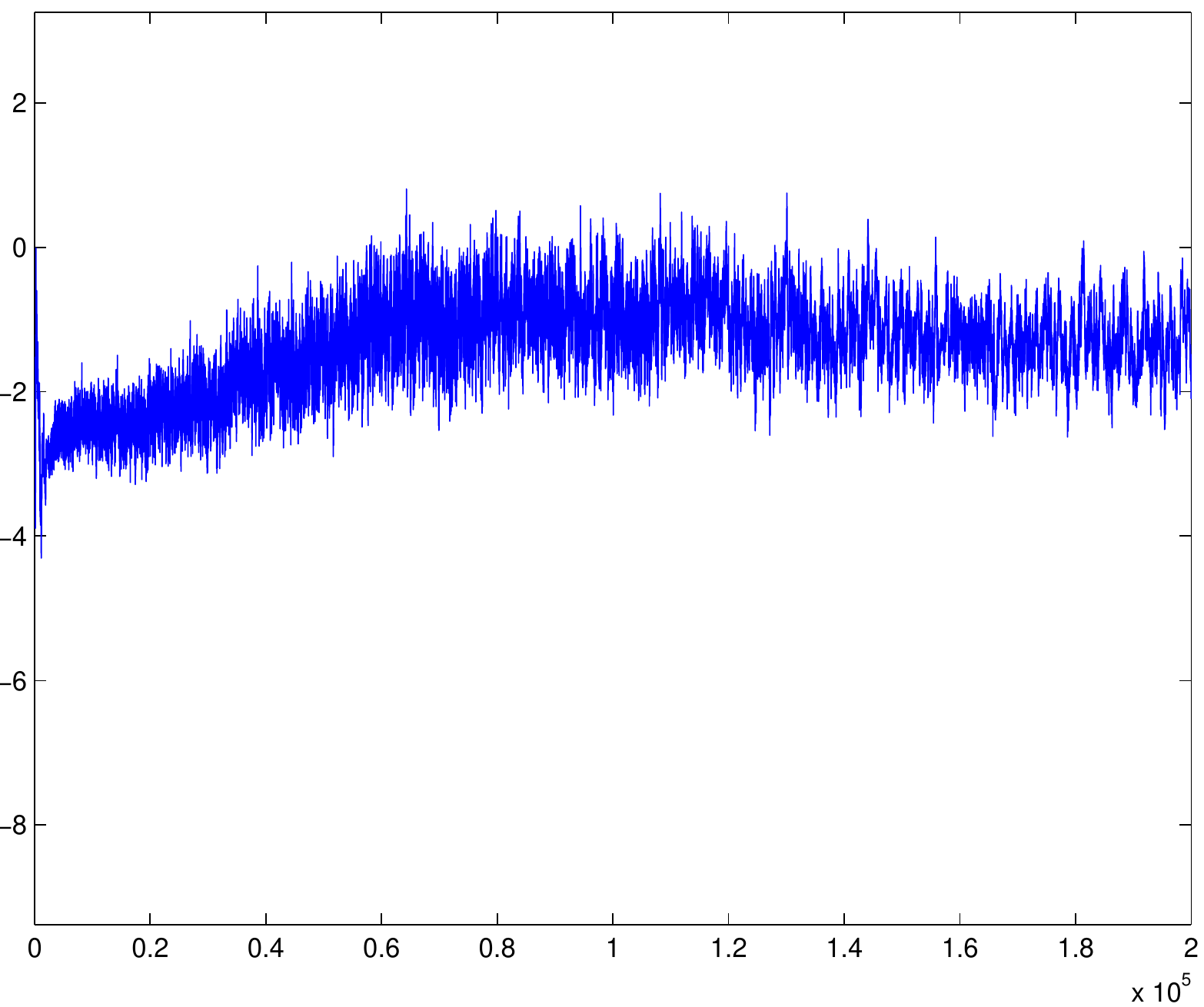}\\
 \end{minipage}
 \begin{minipage}[b]{0.45\linewidth}
 \centering
 (b)
 \includegraphics[width=\linewidth]{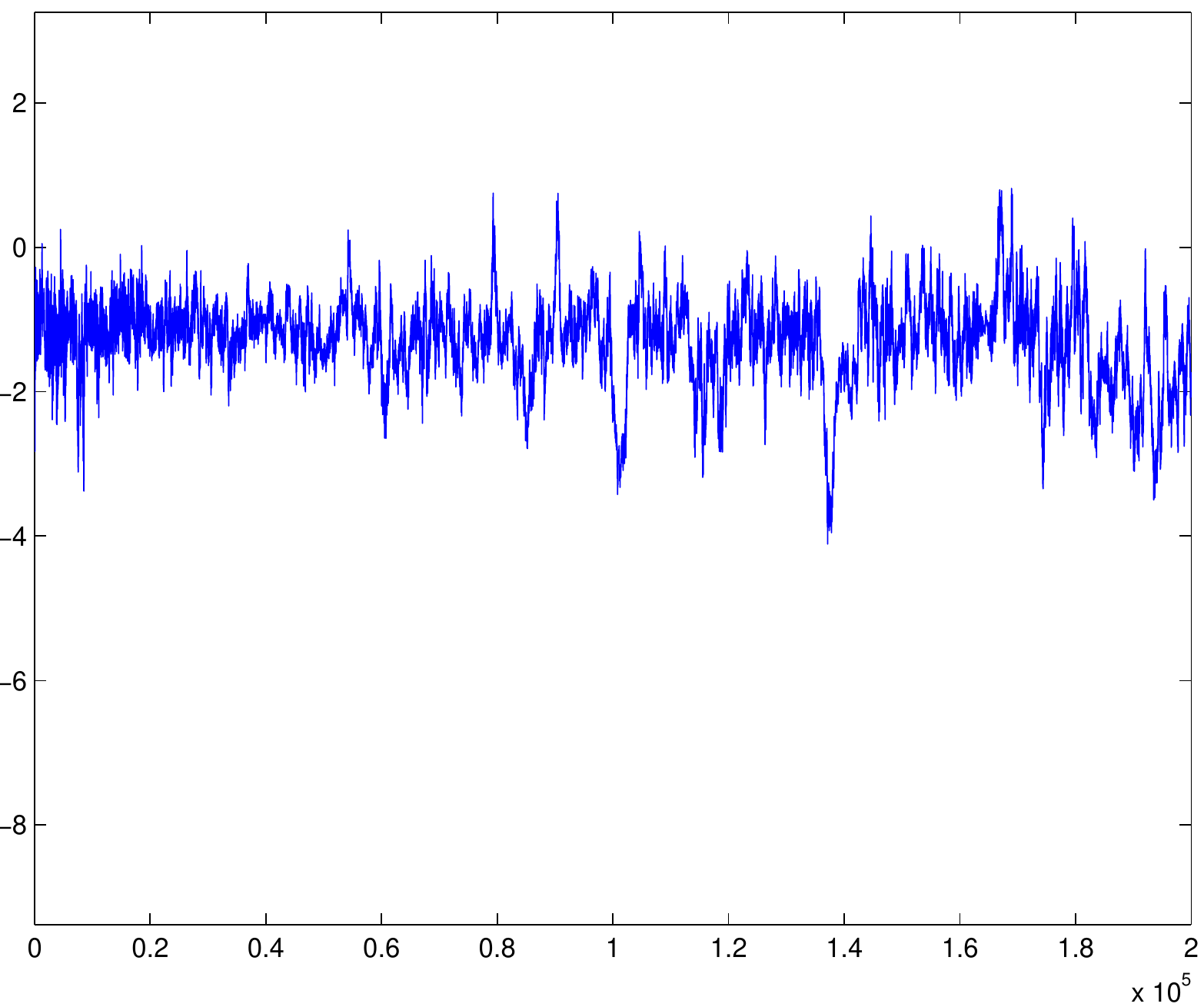}\\
 \end{minipage}\\
 \begin{minipage}[b]{0.45\linewidth}
 \centering
 (c)
 \includegraphics[width=\linewidth]{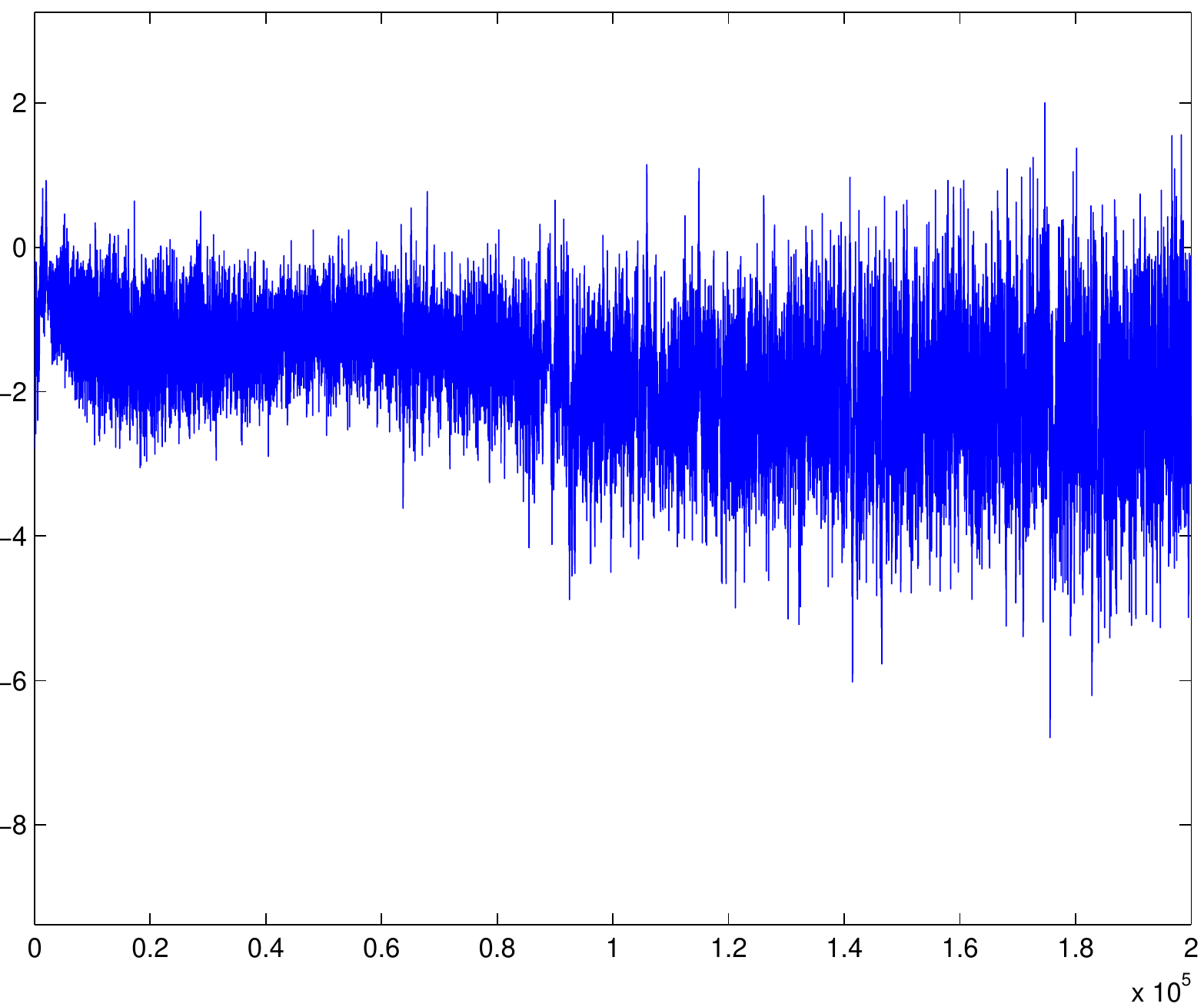}\\
 \end{minipage}
 \begin{minipage}[b]{0.45\linewidth}
 \centering
 (d)
 \includegraphics[width=\linewidth]{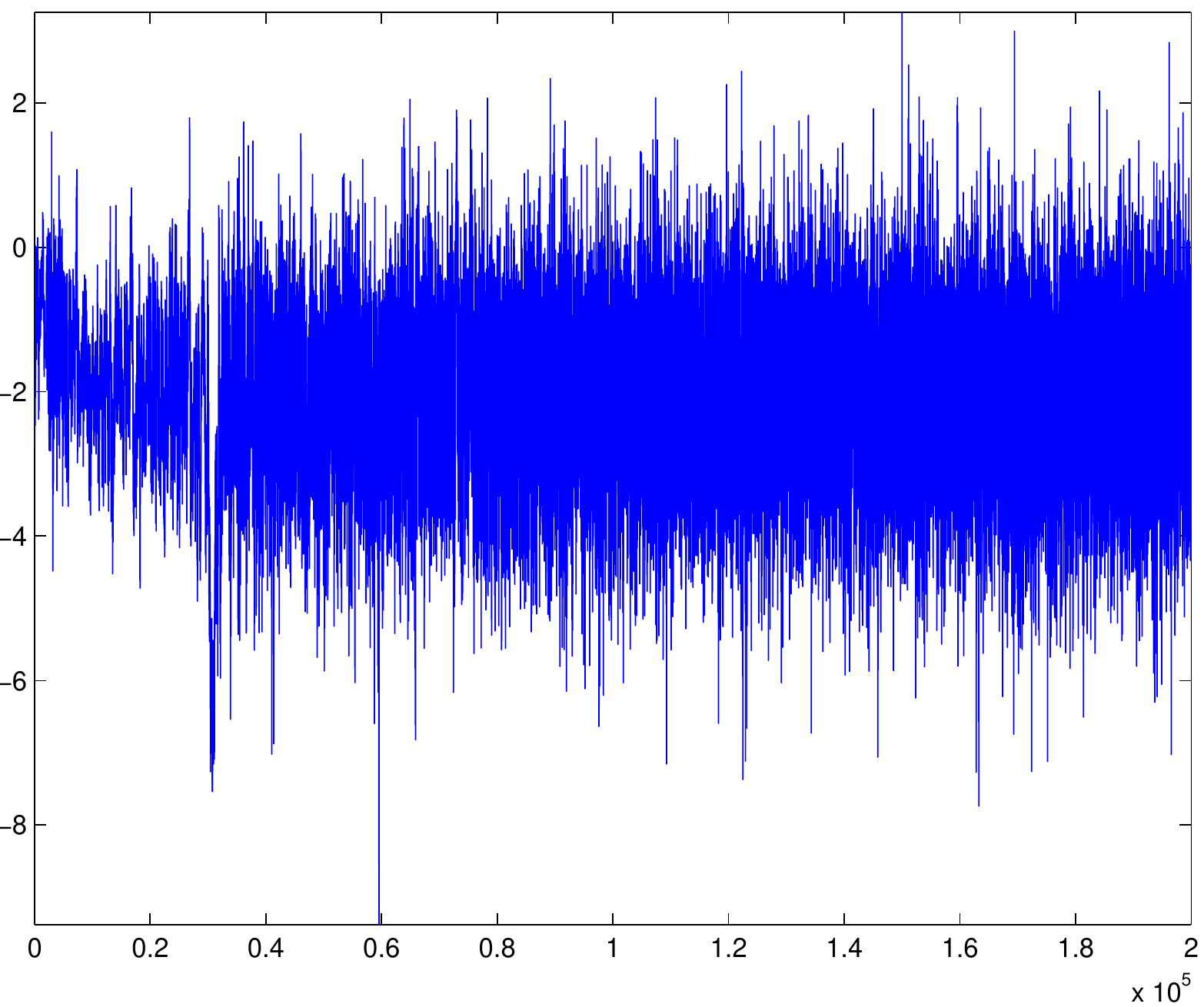}\\
 \end{minipage}
 \end{center}
         \caption{ The panels show trace plots of $x_{21}$ $(X(6.8))$. The samples are taken from runs of ten million MCMC iteration where every fiftieth sample is stored. The methods shown are MHRW with covariance adaptation (Panel a), MHRW with precision adaption (Panel b), MALA with covariance adaption (Panel c), and  MALA with precision adaption (Panel d). }
         \label{fig:trace_X}
  \end{figure}

 \section{Discussion}
 In this article we have derived an adaptive MCMC algorithm that online estimates the Cholesky factor of the precision matrix of a density using a least squares method. The method thus only assumes existence of the second moment of the target distribution and does not rely on any Gaussianity assumptions. We have shown that by taking advantage of the dependency structure of the target density the algorithm can generate an efficient sampling method. The updating of each row in the Cholesky factor is done independently and the method is thus easy to parallelize. Furthermore, the method could be extended so that one updates the Cholesky factor from several independent MCMC trajectories, similarly to the methods by \cite{solonen2012}.


 An interesting research direction is to try to formulate the regression problem for constructing the Cholesky factor such that the formulation varies over the parameter space. For instance, in the first example we used a fixed parameter of the measurement error, $\sigma$, which in most applications instead would be a random variable. Since the smaller $\sigma$ is relative to $\mv{Q}$, the more the posterior distribution of the random field $\mv{x}$ is concentrated around the observations. Thus a proposal for $\mv{x}$ should vary scale with $\sigma$. Which could be constructed by allowing the elements of $\mv{D}$ to vary with $\sigma$.

\section*{Acknowledgements}
Both authors have been supported by the Knut and Alice Wallenberg foundation and the first author has been supported by the Swedish Research Council Grant 2008-5382.

 \begin{appendix}
 \section{Convergence of precision adaption} \label{sec:app_conv}
 From the arguments by \cite{EAMCMC_Roberts}, we need to ensure two conditions to get theoretical justification of using a AMCMC algorithm. The two conditions are \textit{Diminsihing Adapation} and \textit{Bounded convergence}.

 The precision adaptation method satisfies Diminsihing Adapation since the elements updated in Algorithm \ref{alg:Lupdate} only change with $\mathcal{O}(\frac{1}{i})$ at the $i^{th}$ iteration. Bounded convergence is a little more subtle to show and there exists various method to ensure the condition. In the following proposition we link precision adaption with covariance adaption.

 \begin{prop}\label{prop:ContinousQ}
 Let $\mathcal{S}$ be the set of positive definite matrices in $\mathbb{R}^{N \times N}$. Then the \textproc{L-update} function in Algorithm \ref{alg:Lupdate},  is a continuous function from  $\mathcal{S}$ to the set of invertible lower triangular matrices.
 \end{prop}
 \begin{proof}
 To show that  \textproc{L-update} is a continuous function, one only need to show that the operations in line 7 and 8 are continuous, which follows trivially if  $\mv{\Sigma}_{(A_i \cup i),(A_i \cup i)}$ is invertible. That $\mv{\Sigma}_{(A_i \cup i),(A_i \cup i)}$ is invertible follows from the Poincar\'e separation theorem \citep{magnus1988matrix}.

 To show that \textproc{L-update} produces an invertible matrix it is enough to show that $\mv{D}_{jj}>0$ (row 8) for  $j=1,\ldots,N$. This follows from that $\mv{D}_{jj}$ is the Schuur complement of  $\mv{\Sigma}_{A_i ,A_i }$ in the matrix $\mv{\Sigma}_{(A_i \cup i),(A_i \cup i)}$, which is positive definite. \qed
 \end{proof}

 Using Proposition \ref{prop:ContinousQ} it is enough to show that the empirical covariance matrix is in a compact domain of positive definite matrices. So for example, if equation (14) in \citep{haario2001} is satisfied, then the precision adaptation satisfies Diminsihing Adapation.

 \end{appendix}



\bibliographystyle{chicago}
 \bibliography{AdaptiveMH.bib}

\end{document}